\documentclass[letterpaper, 11pt]{article}
\usepackage[utf8]{inputenc}

\usepackage{geometry}
\usepackage{amsmath,amssymb}
\usepackage{amsthm}
\usepackage{amssymb,graphicx,verbatim,enumitem,mdframed,mathrsfs}
\usepackage{mdframed}
\usepackage{subcaption}
\usepackage[toc]{appendix}

\usepackage{algorithm}
\usepackage{algorithmic}
\usepackage{arydshln}

\newtheorem*{theorem*}{Theorem}

\newtheorem{theorem}{\bfseries Theorem}
\newtheorem{definition}{\bfseries Definition}
\newtheorem{corollary}{\bfseries Corollary}
\newtheorem{lemma}{\bfseries Lemma}
\newtheorem{assumption}{\bfseries Assumption}

\providecommand{\mnorm}[1]{\ensuremath{\left\lvert#1\right\rvert}}

\newcommand{\printcomment}[1]{\textcolor{red}{\bf #1}~\\}

\def\R{\mathbb{R}}

\def\C{\mathcal{C}}
\def\G{\mathcal{G}}

\def\F{\mathcal{F}}

\title{Resilience in Collaborative Optimization: \\Redundant and Independent Cost Functions}
\author{Nirupam Gupta \hspace*{.7in} Nitin H. Vaidya\\Georgetown University}
\date{}

\begin{document}

\maketitle

%%%%%%%%%%%%%%%%%%%%%%%%% INTRODUCTION %%%%%%%%%%%%%%%%%%%%%%%%%%%%
\begin{abstract}
     This report considers the problem of Byzantine fault-tolerance in multi-agent collaborative optimization. In this problem, each agent has a local cost function. The goal of a collaborative optimization algorithm is to compute a minimum of the aggregate of the agents' cost functions. We consider the case when a certain number of agents may be Byzantine faulty. Such faulty agents may not follow a prescribed algorithm, and they may send arbitrary or incorrect information regarding their local cost functions. A reasonable goal in presence of such faulty agents is to minimize the aggregate cost of the non-faulty agents. In this report, we show that this goal can be achieved {\em if and only if} the cost functions of the non-faulty agents have a {\em minimal redundancy} property. We present different algorithms that achieve such tolerance against faulty agents, and demonstrate a trade-off between the complexity of an algorithm and the properties of the agents' cost functions. 
     
     Further, we also consider the case when the cost functions are independent or do not satisfy the {\em minimal redundancy} property. In that case, we quantify the tolerance against faulty agents by introducing a metric called {\em weak resilience}. We present an algorithm that attains weak resilience when the faulty agents are in the minority and the cost functions are non-negative.
\end{abstract}

\section{Introduction} 
\label{sec:intro}
The problem of collaborative optimization in multi-agent systems has gained significant attention in recent years~\cite{boyd2011distributed, nedic2009distributed, duchi2011dual, rabbat2004distributed, raffard2004distributed}. In this problem, each agent knows its own {\em local} objective (or cost) function. In the fault-free setting, all the agents are non-faulty (or honest), and the goal is to design a distributed (or collaborative) algorithm to compute a minimum of the aggregate of their local cost functions. We refer to this problem as {\em collaborative optimization}. Specifically, we consider a system of $n$ agents where each agent $i$ has a local real-valued cost function $f_i(x)$ that maps a point $x$ in  $d$-dimensional real-valued vector space (i.e.~$\R^d$) to a real value. Unless otherwise stated, the cost functions are assumed to be convex\footnote{As noted later in Section~\ref{sec:sum}, some of our results are valid even when the cost functions are {\em non-convex}.}~\cite{boyd2004convex}. The goal of collaborative optimization is to determine a {\em global} minimum $x^*$, such that
\begin{align}
    x^* \in \arg \min_{x \in \R^d} ~ \sum_{i = 1}^n f_i(x) . \label{eqn:orig_obj}
\end{align}
Throughout the report, we use the shorthand `$\min$' for `$\min_{x \in \R^d}$', unless otherwise mentioned.\\

% To ensure existence of a {\em global} minimum, the cost functions are assumed to be bounded from below, i.e., $\min_x f_i(x) > -\infty$ for all $i$. Also, the domain of all the cost functions is assumed to be the entire real space $\R^d$.\\
% \printcomment{+++++ check consistency of the above assumption with Assumption 1 +++++++++++ {\color{blue} Omitted. Assumption 1 suffices.}}

% The above equation specifies that $w^*$  equals an argument $w$that minimizes the specified global cost $\sum_{i = 1}^n Q_i(w)$.
As a simple example, $f_i(x)$ may denote the cost for an agent $i$ (which may be a robot or a person) to travel to location $x$ from its current location. In this case, $x^*$ is a location that minimizes the total cost for all the agents. Such multi-agent collaborative optimization 
is of interest in many practical applications,
 including collaborative machine learning~\cite{bottou2018optimization, boyd2011distributed, kairouz2019advances}, swarm robotics~\cite{raffard2004distributed}, and collaborative sensing~\cite{rabbat2004distributed}. Most of the prior work assumes all the agents to be non-faulty. Non-faulty agents follow a specified algorithm correctly. In our work we consider a scenario wherein some of the agents may be faulty and may behave incorrectly. \\

Su and Vaidya \cite{su2016fault} introduced the problem of
collaborative optimization in the presence of a Byzantine faulty agents. A Byzantine faulty agent may behave arbitrarily~\cite{lamport1982byzantine}. In particular, the faulty agents
may send incorrect and inconsistent information in order to bias the output of a collaborative optimization algorithm, and the faulty agents may also collaborate with each other. For example, consider an application
of multi-agent collaborative optimization to the case of collaborative sensing where the agents (or {\em sensors}) are observing a common {\em object} in order to collectively identify the object. However, the faulty agents may send arbitrary observations concocted to prevent the non-faulty agents from making the correct identification~\cite{chen2018resilient, chong2015observability, pajic2014robustness}. Similarly, in the case of collaborative learning, which is another application of multi-agent collaborative optimization, the faulty agents may send incorrect information based on {\em mislabelled} or arbitrary concocted data points to prevent the non-faulty agents from learning a {\em good} classifier~\cite{alistarh2018byzantine, bernstein2018signsgd, blanchard2017machine, charikar2017learning, chen2017distributed, xie2018generalized}. 

\subsection{System architecture}
\label{sub:arch}

The contributions of this paper apply to two different system architectures illustrated in Figure \ref{fig:sys}. In the server-based architecture, the server is assumed to be trustworthy, but up to $t$ agents may be Byzantine faulty. The trusted server helps solve the distributed optimization problem in coordination with the agents. In the peer-to-peer architecture, the agents are connected to each other by a complete network, and up to $t$ of these agents may be Byzantine faulty.
Provided that $t<\frac{n}{3}$, any algorithm for the server-based architecture can be simulated in the peer-to-peer system using
the well-known {\em Byzantine broadcast} primitive \cite{lynch1996distributed}.\\

For the simplicity of presentation, the rest of this report assumes the server-based architecture.\\

\begin{figure}[htb!]
\centering
\centering \includegraphics[width=0.6\textwidth]{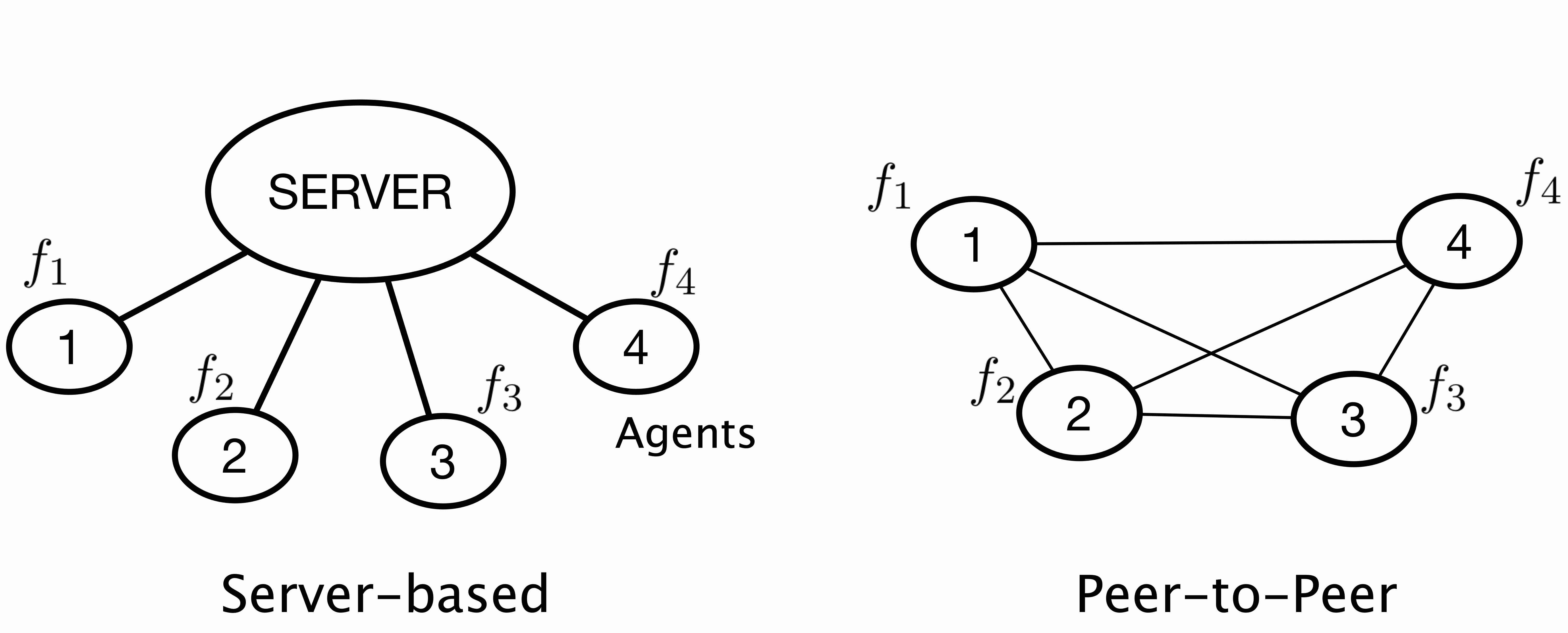} 
\caption{The system architectures.}
\label{fig:sys}
\end{figure}

\subsection{Resilience in collaborative optimization}
\label{sub:res}

As stated above, we will assume the server-based architecture in the rest of our discussion.
We assume that up to $t$ of the $n$ agents may be Byzantine faulty, such that $n>2t$.\\

We assume that each agent $i$ has a ``true'' cost function. Unless otherwise noted, each such cost function is assumed to be convex.
\begin{itemize}
\item If an agent $i$ is non-faulty, then its behavior is consistent with its true cost function, say $g_i(x)$. For instance,
if agent $i$ is required to send to the server the value of its cost function at some point $x^\dagger$, then a non-faulty agent $i$ 
will indeed send $g_i(x^\dagger)$. 

\item If an agent $i$ is faulty, then its behavior can be arbitrary, and not necessarily
consistent with its true cost function, say $f_i(x)$.
For instance, if agent $i$ is required to to send to the server the value of its cost function at some point $x^\dagger$, then a faulty agent $i$ may send an arbitrary value instead of $f_i(x^\dagger)$.

Clearly, when an agent is faulty, it may not share with the server correct information about its true cost function.
However, it is convenient to define its true cost function as above, which is the cost function it would use in the absence of its failure.

% then its behavior may not necessarily be consistent with its true cost function $f_i(x)$.may send information corresponding to some arbitrary (or {\em bogus}) cost function $h_i$ instead of its {\em true} cost function $f_i$. On the other hand,
\end{itemize}

Throughout this report, we assume the existence of a finite minimum for the aggregate of the true cost functions of the agents. Otherwise, the objective of collaborative optimization is vacuous. Specifically, we make following technical assumption.

\begin{assumption}
\label{asp:basic}
Suppose that the true cost function of each agent $i$ is $f_i(x)$. Then, for every non-empty set of agents $T$, we assume that there exists a finite $x^*$ such that $x^* \in \arg\min_{x \in \R^d} \, \sum_{i \in T} f_i(x)$.
\end{assumption}

Suppose that the true cost function of agent $i$ is $f_i(x)$.
Then, ideally, the goal of collaborative optimization is to compute
a minimum of the aggregate of the {\em true} cost functions of all the
$n$ agents, $\sum_{i = 1}^n f_i(x)$, even if some of the agents are Byzantine faulty. In general, this may not feasible since the Byzantine faulty agents can behave arbitrarily. To understand the feasibility of achieving some degree of resilience to Byzantine faults, we consider two cases.

\begin{itemize}
\item {\em Independent functions:} 
A set of cost functions are \textit{independent} if information about some of the functions in the set does not help learn any information about the remaining functions in the set. In other words, the cost functions do not contain any redundancy.
\item {\em Redundant functions:}
Intuitively speaking, a set of cost functions includes \textit{redundancy} when knowing some of the cost functions helps to learn some information about the remaining cost functions. 
As a trivial example, consider the special case when it is known that there exists some function $g(x)$  such that $g(x)$ is the true cost function of every agent.
In this case, knowing the true cost function of any agent suffices to learn the true cost functions of all the agents. Also,
any $x$ value that minimize an individual agent's true cost function also minimizes the total true cost over all the agents.
% $$\arg \min_{x} ~ \sum_{i \in S} f_i(x) = \arg \min_{x} ~  (n\,g(x)) = \arg \min_{x} f_j(x),~\forall j\in S$$
% Thus, each agent can use a deterministic algorithm to identically compute a point that minimizes its local cost function, and achieve $t$-resilience.
%
% In general, different forms of redundancy may be present in the cost functions. In this report, we consider $2t$-redundancy defined soon.
\end{itemize}

%
% In this case, minimizing the aggregate of the cost functions of {\em all} the agents is not a desirable goal, because this aggregate will include cost functions of the faulty agents as well, and the faulty agents can potentially {\em choose} an arbitrary cost function arbitrarily. 

% Traditional collaborative optimization algorithms are vulnerable against faulty agents~\cite{blanchard2017machine, charikar2017learning, chong2015observability, pajic2017attack}. What this means is that in case a certain are faulty, then they can arbitrarily bias the output of a collaborative optimization algorithm. Instances of such if a certain agents are faulty, then i.e~they may report incorrect objective functions either inadvertantly or  a collaborative optimization is highly . In this report, we consider a scenario where up to $t$ (less than $n$) of the agents' objective functions, fed as inputs to a collaborative optimization algorithm, are corrupted by a {\em Byzantine} adversary~\cite{lamport1982byzantine}. 
% \begin{itemize}
%     \item 
%     \item 
% \end{itemize}
% is assumed omniscient, i.e.  , and . Meaning that we may not be able to differentiate a corrupted objective function from a non-corrupted objective function. For instance, if we were to assume that the non-corrupted objective functions are convex and continuous then the adversary would choose a convex and continuous objective functions to remain undetected. In short, 
% However, the value of $t$ is assumed to be known.\\
% \subsection{Resilience under redundancy}
% \label{sub:red}

Su and Vaidya \cite{su2016fault} defined the goal of { fault-tolerant} collaborative optimization as minimizing the 
aggregate of cost functions of just the non-faulty agents. Specifically, if $f_i(x)$ is the true cost function of agent $i$, and $S$ denotes the set of non-faulty agents in a given execution,
then they defined the goal of fault-tolerant optimization to be to output a point in
\begin{align}
    \arg \min_{x \in \R^d} ~ \sum_{i \in S} f_i(x). \label{eqn:hon_obj}
\end{align}
We refer to the above goal as \textit{$t$-resilience}, formally defined below.\\

\noindent \fbox{\begin{minipage}{0.98\textwidth}
\begin{definition}[{\em $t$-resilience}]
\label{def:t_res}
A collaborative optimization algorithm is said to be \textit{$t$-resilient} if it outputs a minimum of the aggregate of the true cost functions of the \textit{non-faulty} agents despite up to $t$ agents being Byzantine faulty. 
\end{definition}
\end{minipage}}
\\~\\

% Note that, by definition, a {\em $t$-resilience} collaborative optimization algorithm is deterministic.\\

In general, Su and Vaidya \cite{su2016fault} showed that, because the identity of the faulty agents is a priori unknown, a {$t$-resilient} algorithm may not necessarily exist. 
% 
% {\color{blue}
% We note that for many known applications of collaborative multi-agent optimization, such as collaborative multi-sensing, machine learning and swarm robotics, the agents' cost functions are convex~\cite{bottou2018optimization, boyd2011distributed, rabbat2004distributed, raffard2004distributed, nedic2009distributed}.
% 
% \begin{definition}
% \label{def:convex}
% A functions $f: \R^d \to \R$ is convex if for all $x, ~ y$ in the domain of $f$, and $\theta$ subject to $0 \leq \theta \leq 1$, we have
% \begin{align*}
%     f(\theta x + (1 - \theta) y) \leq \theta f(x) + (1 - \theta) f(y).
% \end{align*}
% \end{definition}
% 
In this report, we provide an exact characterization of the condition
under which {$t$-resilience} is achievable. In particular, we show that $t$-resilience is achievable {if and only if}
the agents satisfy a property named {$2t$-redundancy}, defined next.\footnote{\label{footnote:k} The notion of $2t$-redundancy can be extended to $k$-redundancy by replacing $n-2t$ in Definitions \ref{def:2t_red_alt} and \ref{def:2t_red} by $n-k$.} The definitions below are vacuous if $n \leq 2t$. Henceforth, we assume that the maximum number of faulty agents $t$ are in the minority, i.e., $n > 2t$.

\begin{definition}[{\em $2t$-redundancy}]
\label{def:2t_red_alt}
Let $f_i(x)$ denote the true cost function of agent $i$.
The $n$ agents are said to satisfy \mbox{ $2t$-redundancy} if the following holds for every two subsets
$S_1$ and $S_2$ each containing $n-2t$ agents.
\begin{align}
    \emptyset~\neq~\bigcap_{i \in S_1} \arg \min_{x \in \R^d} f_i(x)~=~ \bigcap_{i \in S_2} \arg \min_{x \in \R^d} f_i(x)\label{def_1}
\end{align}
\end{definition}
The above definition of $2t$-redundancy is equivalent to the definition below, as shown in Appendix~\ref{app:equiv_def}.

\begin{definition}[{\em $2t$-redundancy}]
\label{def:2t_red}
Let $f_i(x)$ denote the true cost function of agent $i$.
The $n$ agents are said to satisfy \mbox{ $2t$-redundancy} if the following holds for any sets of agents $\widehat{S}$ and $S$ such that $\mnorm{S} \geq n-t$, $\mnorm{\widehat{S}} \geq n-2t$, and $\widehat{S} \subseteq S$.
\begin{align}
    \bigcap_{i \in \widehat{S}} \arg \min_{x \in \R^d} f_i(x) = \arg \min_{x \in \R^d} \, \sum_{i \in S}  f_i(x) \label{def_2}
\end{align}
\end{definition}

% \printcomment{++++ Definition 3: For any set S and S-hat such that ... ???? ++++++ check that we use Definition 3 consistent with the new text above ++++++++ for instance, in the proof of equivalence of Definition 2 and 3 ++++ {\color{blue}Noted, and corrections made.}}

Note that the $t$-resilience property pertains the point in $\R^d$ that is the output of a collaborative optimization
algorithm. $t$-resilience property does not explicitly
impose any constraints on the {\em function value}.
% In this report, we also consider the case when the cost functions of the non-faulty agents are {\em independent}, i.e., information about the cost functions of a subset of non-faulty agents does not provide any information about the cost functions of the remaining non-faulty agents.
The notion of {\em $(u, \,t)$-weak resilience} stated below relates to function values.\\

\noindent \fbox{\begin{minipage}{0.98\textwidth}
\begin{definition} [{\em $(u, t)$-weak resilience}]
\label{def:weak_res}
Let $f_i(x)$ denote the true cost function of agent $i$.
Let $S$ denote the set of all non-faulty agents. For $0 \leq u \leq \mnorm{S}$, a collaborative optimization algorithm is said to be {\em $(u, t)$-weak resilient} if it  outputs a point $\widehat{x}$ for which there exists a subset $\widehat{S}$ of $S$ such that $\mnorm{\widehat{S}} \geq \mnorm{S} - u$, and
\begin{align}
    \sum_{i \in \widehat{S}}f_i(\widehat{x}) \leq \min_{x \in \R^d} ~ \sum_{i \in S} f_i(x) ~ . \label{eqn:weak_res}
\end{align}
% and the size of the largest such subset $S$ is $\mnorm{\C^{-}} - s$.
\end{definition}
\end{minipage}}
\\~\\

It can be shown easily that {$(0, t)$-weak resilience} implies { $t$-resilience}. The proof is deferred to Section~\ref{sec:ind}.
In many applications of multi-agent collaborative optimization, such as distributed machine learning, distributed sensing or hypothesis testing and swarm robotics, the cost functions are non-negative~\cite{bottou2018optimization, boyd2011distributed, kairouz2019advances, rabbat2004distributed, raffard2004distributed}. We constructively show that if the true cost functions of the agents are non-negative then {\em $(u, t)$-weak resilience} for $u \geq t$ can be achieved even if the cost functions are independent.

\subsection{Prior Work}
The prior work on resilience in collaborative multi-agent optimization by Su and Vaidya, 2016~\cite{su2016fault}, and Sundaram and Gharesifard, 2018~\cite{sundaram2018distributed}, only consider the special class of {\em univariate} cost functions, i.e, dimension $d$ equals one. On the other hand, we consider the general class of {\em multivariate} cost functions, i.e., $d$ can be greater than one. Specifically, they have proposed algorithms that output a minimum of the {\em non-uniformly} {weighted} aggregate of the non-faulty agents' cost functions when $d = 1$. However, their proposed algorithms do not extend easily for the case when $d > 1$. On the other hand, the algorithms and the fault-tolerance results presented in this report are valid regardless of the value of the dimension $d$ as long as it is finite.\\

Su and Vaidya have also considered a special case where the true cost functions of the agents are convex combinations of a finite number of basis convex functions in~\cite{su2016robust}. They have shown that if the basis functions have a common minimum then a minimum point (as in \eqref{eqn:hon_obj}) can be computed accurately. This property of redundancy in the minimum of the basis functions, we note, is a special case of the {\em $2t$-redundancy} property that we prove necessary and sufficient for {\em $t$-resilience} in this report. Other prior work related to the {\em 2t-redundancy} property is discussed in Section~\ref{sub:red_prior}.\\

Yang and Bajwa, 2017~\cite{yang2017byrdie} consider a very special case of collaborative optimization problem. They assume that the {\em multivariate} cost functions that can be {\em split} into independent {\em univariate} {\em strictly convex} functions. For this special, they have extended the fault-tolerance algorithm of Su and Vaidya, 2016~\cite{su2016fault} for approximate resilience. In general, however, the agents' cost functions do not satisfy such specific properties. In this report, we do not make such assumptions about the agents' cost functions. We only assume the cost functions to be convex, differentiable and that the minimum of their sum is finite (i.e., Assumption~\ref{asp:basic}). Note that these assumptions are fairly standard in the optimization literature, and are also assumed in all of the aforementioned prior work. \\

% \printcomment{++++ Lili's work also considered some type of redundancy ++++ {\color{blue} noted above }}

{\bf Outline of the report:} The rest of the report is organized as follows. In Section~\ref{sec:thm_red}, we present the case when the cost functions have redundancy. In Section~\ref{sec:ind}, we present the case when the cost functions are independent. In Section~\ref{sec:norm}, we summarize a gradient-based algorithm for {\em $t$-resilience}, which was proposed in our prior work~\cite{gupta2019byzantine}. In Section~\ref{sec:sum}, we discuss direct extension of our results to the case when the cost functions are {\em non-differentiable} and {\em non-convex}. In the same section, we also present a summary of our results.

\section{The Case of Redundant Cost Functions}
\label{sec:thm_red}

This section presents the key result of this report for the case when the cost functions are redundant. Unless otherwise mentioned, in the rest of the report, the cost functions are assumed to be differentiable, i.e., their gradients exist at all the points in $\R^d$. Indeed, the cost functions are differentiable for most aforementioned applications of collaborative optimization~\cite{bottou2018optimization, boyd2011distributed, rabbat2004distributed, raffard2004distributed}. Nevertheless, as elaborated in Section~\ref{sec:sum}, some of our results are also applicable for non-differentiable cost functions.\\

% \printcomment{+++++ it should be enough to assume finite min for each true cost function, since the min for a sum is in the convex hull of the min of the individual cost functions {\color{blue} is this true? Nitin: +++++ not true, not quite in the convex hull, but I think we can show some relationship to the minimum of each function -- for instance, the minimum of the sum may be in the box that contains all the points ++++ that is, find the convex hull along each dimension and take the box with the convex hull in each dimension as one side} ++++++++++ it would be better to make the weaker assumption, and derive the above condition for T as an implication +++++++ }

% \printcomment{+++++ I guess to save time we can leave Assumption 1 as it is. We can simplify it later in a future revision. Anyway, the main contributions are not affected by how this assumption is stated. ++++++ {\color{blue} agree.}}

% \printcommentNiru{ +++ Nirupam will make necessary changes in the proofs ++++ }

Before we present Theorem~\ref{thm:imp} below which states the key result of this section, in Lemma \ref{lem:equiv_red} we present an alternate, and perhaps more natural, equivalent condition of the {\em $2t$-redundancy} property for the specific case when the agents' cost functions are differentiable.
The proof of Lemma \ref{lem:equiv_red} uses Lemma \ref{lem:non-empty-x} stated below.

\begin{lemma}
\label{lem:non-empty-x}
Suppose that Assumption~\ref{asp:basic} holds true, and $n > 2t$. For a non-empty set $T$, consider  a set of functions $g_i(x)$, $i\in T$, such that $$\bigcap_{i\in T}\arg\min_x g_i(x)\neq \emptyset.$$ Then
$$\bigcap_{i\in T}\arg\min_x g_i(x) = \arg\min_x \sum_{i\in T}g_i(x).$$
\end{lemma}
Appendix \ref{append:non-empty-x} presents the proof of the above lemma.

\begin{lemma}
\label{lem:equiv_red}
Suppose that Assumption~\ref{asp:basic} holds true, and $n> 2t$. When the true cost functions of the agents are convex and differentiable then the {\em $2t$-redundancy} property stated in Definition~\ref{def:2t_red_alt} or Definition~\ref{def:2t_red} is equivalent to the following condition:
\begin{itemize}
    \item[] A point is a minimum of the sum of true cost functions of the non-faulty agents {\em if and only if} that point is a minimum of the sum of the true cost functions of any $n-2t$ non-faulty agents.
\end{itemize}
\end{lemma}
\begin{proof}
% For $t=0$, the proof is trivial by Assumption~\ref{asp:basic}. Thus, we assume $t>0$ in the remaining proof. 
Let the true cost function of each agent $i$ be denoted by $f_i(x)$. Recall that there can be at most $t$ Byzantine faulty agents. Let $S$ with $\mnorm{S} \geq n-t$ be the set of the non-faulty agents.
~\\

{\bf Part I:} We first show that the condition stated in the lemma implies that in Definition~\ref{def:2t_red_alt}. Recall that the conditions in Definitions~\ref{def:2t_red_alt} and~\ref{def:2t_red} are equivalent.\\

The condition stated in the lemma is equivalent to saying that for every subset $\widehat{S}$ of $S$ of size $n-2t$,
\begin{align}
    \arg \min ~ \sum_{i \in \widehat{S}} f_i(x) = \arg \min ~ \sum_{i \in S} f_i(x). \label{eqn:sum_sum}
\end{align}
We show below that~\eqref{eqn:sum_sum} together with Assumption~\ref{asp:basic} imply that for every subset $\widehat{S}$ of $S$ of size $n-2t$, 
\begin{align}
    \bigcap_{i \in \widehat{S}} \arg \min f_i(x) \neq \emptyset. \label{eqn:imply_sum_sum}
\end{align}

Consider two arbitrary agents $i, ~j$ in $S$, and then consider two size $(n-2t)$ subsets $S_i$ and $S_j$ of $S$ such that 
$i\in S_i$, $j\in S_j$, and
\begin{align}S_i \setminus \{i\} = S_j \setminus \{j \}. \label{eq_ij}\end{align}
By Assumption~\ref{asp:basic}, there exists a point $x^* \in \arg \min \sum_{i \in S}f_i(x)$. Now,~\eqref{eqn:sum_sum} implies that
\begin{align*}
    \nabla ~ \sum_{l \in S_i} f_l (x^*) = \nabla ~ \sum_{l \in S_j} f_l (x^*)  = 0.
\end{align*}
The above equality and \eqref{eq_ij} imply that
\[ \nabla f_i(x^*) = \nabla f_j(x^*)\]
This equality can be proven for any $i,j\in S$.
As the true cost functions $f_1, \ldots, ~ f_n$ are assumed convex, from above we obtain,
\begin{align*}
    x^* \in \arg \min_x f_i(x), \quad \forall \, i \in S.
\end{align*}
Therefore, for every subset $\widehat{S}$ of $S$ of size $n-2t$,
\[x^* \in \bigcap_{i \in \widehat{S}} \arg \min_x f_i(x) \neq \emptyset.\]
The above implies that for every subset $\widehat{S}$ of $S$ of size $n-2t$,
\[\arg \min ~ \sum_{i \in \widehat{S}} f_i(x) =  \bigcap_{i \in \widehat{S}} \arg \min f_i(x).\]
The above together with~\eqref{eqn:sum_sum} implies the condition in Definition~\ref{def:2t_red_alt}, i.e., 
\[\bigcap_{i \in S_1} \arg \min f_i(x) = \bigcap_{i \in S_2} \arg \min f_i(x), \quad \forall \, S_1, ~ S_2 \subset S, ~ \mnorm{S_1} = \mnorm{S_2} = n-2t.\]
~\\

{\bf Part II:}
We now show that the condition in Definition \ref{def:2t_red} implies the condition stated in the lemma. Now, $\arg \min \sum_{i \in S} f_i(x)$ (i.e., the right side of (\ref{def_2})) is a non-empty set due to Assumption \ref{asp:basic}. This and (\ref{def_2}) imply that for every subset $\widehat{S} \subset S$ of size $n-2t$, 
\[\bigcap_{i \in \widehat{S}} \arg \min f_i(x) \neq \emptyset.\]
Therefore, by Lemma \ref{lem:non-empty-x},
\[\bigcap_{i \in \widehat{S}} \arg \min f_i(x) = \arg \min ~ \sum_{i \in \widehat{S}} f_i(x).\]
Substituting the above in~\eqref{def_2} implies~\eqref{eqn:sum_sum} which is equivalent to the condition stated in the lemma. 
\end{proof}
~

The following theorem presents the main result of this section.

\def\out{\mathsf{Out}}
\def\d{\partial}
\begin{theorem}
\label{thm:imp}
Suppose that Assumption~\ref{asp:basic} holds true, and $n > 2t$.
When the true cost functions of the agents are convex and differentiable then {\em $t$-resilience} can be achieved if and only if the agents satisfy the \mbox{\em $2t$-redundancy} property.
\end{theorem}

% \printcomment{ ++++++++ Move the necessity proof to the Appendix ++++++ keep the algorithm and its proof in main body ++++++}

\begin{proof}
The case of $t$=0 is trivial, since there are no faulty agents.
In the rest of the proof, we assume that $t\geq 1$. \\

{\bf Sufficiency of $2t$-redundancy:} Sufficiency of $2t$-redundancy is proved constructively using the
algorithm presented in Section \ref{sub:t_algo}. In particular, the algorithm
is proved to achieve $t$-resilience if $2t$-redundancy holds.\\

{\bf Necessity of $2t$-redundancy:} 
We consider the worst-case scenario where $t$ arbitrary agents are faulty. Suppose that $t$-resilience can be achieved using an algorithm named $\Pi$. Consider an execution $E_S$ of $\Pi$ in which set $S$ with $\mnorm{S} = n-t$ is the actual set of non-faulty agents. All the remaining agents in the set $C = \{1, \ldots, \, n\} \setminus S$ are the actual faulty agents. Suppose that the true cost function of agent $i$ in execution $E_S$ is $g_i(x)$. We assume that the functions $g_1, \ldots, \, g_n$ are differentiable and convex. \\
% \printcommentNiru{ ++++++++ Do we need the next line? ++++++} We assume that functions $g_1(x), g_2(x),\cdots, g_n(x)$ are
% independent of each other.\\

In any $t$-resilient algorithm for collaborative optimization, the server can communicate with the agents and learn some information about their local cost functions. The most information the server can learn about the cost function of an agent $i$ is the complete description of its local cost function. To prove the necessity of $2t$-redundancy, we assume that the server knows a cost function reported by each non-faulty agent $i$.

Now consider the following executions.
\begin{itemize}

\item In execution $E_0$, all the agents are non-faulty. Let $S_0$ denote the set of all agents, which happen to be non-faulty in execution $E_0$. Thus, $S_0=\{1,2,\cdots,n\}$. The true cost function of agent $i$ is $g_i(x)$, identical to its
true cost function in execution $E_S$.

% The server will correctly learn the $n$ cost functions for the $n$ agents (which are all non-faulty in $E_0$).

\item In execution $E_i$, where $1\leq i\leq n$, agent $i$ is Byzantine faulty, and all the remaining $n-1$ agents are non-faulty. Let $S_i=S_0\setminus \{i\}$ denote the set of agents that happen to be non-faulty in execution $E_i$.
In execution $E_i$, the true cost function of each non-faulty agent $i$ is $g_i(x)$, which is identical to its true cost function
in execution $E_S$. 
Let the true cost function of
faulty agent $i$ in execution $E_i$ be a differentiable and convex function $h_i(x)$.
Assume that the functions $g_j(x)$, $\forall j$, and $h_i(x)$ are independent.
In  execution $E_i$, suppose that the {\em behavior} of faulty agent $i$ from the viewpoint of the server is consistent with the cost function  $g_i(x)$ (which equals the true cost function of agent $i$ in execution $E_0$).

\end{itemize}

Fix a particular $i$, $1\leq i\leq n$.
From the viewpoint of the server, execution $E_0$ and execution $E_i$ are indistinguishable.
Thus, the $t$-resilient algorithm $\Pi$ will produce an identical output in these executions; suppose that this output is $x_\Pi$.
% 
% We fix the set of cost functions and denote them by $h_1(x), \ldots, \, h_n(x)$. Recall that if agent $i$ is non-corrupted then $h_i(x) \equiv f_i(x)$, otherwise $h_i(x)$ is an arbitrary convex function. To prove the impossibility result, we assume the best-case scenario where a collaborative optimization algorithm has complete information about the cost functions $h_1(x), \ldots, \, h_n(x)$. In the general setting, a collaborative optimization algorithm may only have partial information about the cost functions. Therefore, if {\em $t$-resilience} is impossible in the best-case scenario then it will be impossible in the general setting too.\\
%
% Suppose that an algorithm named $\Pi$ is a {\em $t$-resilient} collaborative optimization algorithm, and let $x_\Pi$ denote the output of $\Pi$. Consider two possible executions of $\Pi$, namely execution-1 and execution-2. In execution-1, the agents in set $S_1 = \{1, \ldots, \, n\}$ are non-faulty, and in execution-2, the agents in set $S_2 = \{2, \ldots, \, n\}$ are non-faulty. In execution-1, all agents are non-faulty, and in execution-2, agent 1 is faulty. In both these executions, the  number of faulty agents is $\leq t$, since $t\geq 1$.  As the corrupted cost functions are assumed convex, and the identity of the corrupted agents is a priori unknown, $\Pi$ cannot differentiate between the two executions and outputs $x_{\Pi}$ for both the executions.
% 
% \printcomment{+++++++++ in the next argument below, we only need 1-resilience, not necessarily $t$-resilience?? +++++++++++++}
As $\Pi$ is assumed to be {\em $t$-resilient}, we have by Definition \ref{def:t_res} and Assumption~\ref{asp:basic},
\begin{align}
\begin{split}
    x_{\Pi} & \in \arg \min \sum_{j \in S_0} g_j(x), \quad \text{ and } \label{eqn:out_pi}\\
    x_{\Pi} & \in \arg \min \sum_{j \in S_i} g_j(x)
\end{split}
\end{align}
For a differentiable cost function $g: \R^d \to \R$, we denote its gradient at a point $x$ by $\nabla g(x)$. Let $\bf{0}$ denote the zero-vector of dimension $d$. If $x^* \in \arg \min_x g(x)$ then
\begin{align}
    \nabla g(x^*) = {\bf 0}. \label{eqn:fact_convex}
\end{align}
% For a function $f: \R^d \to \R$, we denote its set of subgradients at a point $x$ by $\d f(x)$. Specifically, if we denote the inner product between two vectors $u$ and $v$ in $\R^d$ by $\iprod{u}{v}$, then
% \[\d f(x) = \left\{ v \in \R^d ~ \vline ~ f(y) \geq f(x) + \iprod{v}{y - x} \right\}.\]
% Let $\bf{0}$ denote the zero-vector of dimension $d$. 
% % Consider a convex function $f: \R^d \to \R$. 
% If $x^* \in \arg \min f(x)$ then
% \begin{align}
%     {\bf 0} \in \d f(x^*). \label{eqn:fact_convex}
% \end{align}
% The Minkowski sum of sets $A$ and $B$ is defined as follows:
% \[A+B = \{a+b~|~a\in A,~b\in B\}\]
% Recall that $\d f(x)$ denotes the set of subgradients of $f(x)$.
% For two convex functions $f(x): \R^d \to \R$ and $g(x): \R^d \to \R$ the sum of their subgradients is equal to the subgradient of their sum~\cite{boyd2004convex}, i.e.,
% \printcomment{++++++++++ do the functions have to be convex for this? see previous sentence ++++++++++ [\color{blue}Yes. In general, the sum of their subgradients is a subset of the subgradient of their sum.]}
% \begin{align}
%     \d f(x) + \d g(x) = \left\{ v + w ~ \vline ~ v \in \d f(x), ~ w \in \d g(x)\right\} = \d (f + g)(x), \quad \forall x.
%  \label{eqn:fact_convex_2}
% \end{align}
% \printcomment{+++++++++++++++++++ summing subgradients is summing sets ++++ need to articulate this? ++++++++++++any exceptions? ++++++++}
% As the cost functions of the non-faulty agents are assumed convex,
Form~\eqref{eqn:out_pi} and~\eqref{eqn:fact_convex} we obtain,
\begin{align}
\begin{split}
    \nabla \sum_{j \in S_0} g_j(x_{\Pi}) = \sum_{j \in S_0} \nabla g_j(x_\Pi) & = {\bf 0}, \text{ and } \\
    \nabla \sum_{j \in S_i} g_j(x_{\Pi}) = \sum_{j \in S_i} \nabla g_j(x_{\Pi}) &= {\bf 0}. \label{eqn:out_grad_1}
\end{split}
\end{align}
Recall that $S_0 = S_i \cup \{i\}$. Therefore,
\begin{align}
      \nabla g_i(x_\Pi) + \nabla \sum_{j \in S_i} g_j(x_{\Pi}) = \nabla \sum_{j \in S_0} g_j(x_{\Pi}). \label{eqn:out_grad_2}
\end{align}
From~\eqref{eqn:out_grad_1} and~\eqref{eqn:out_grad_2} we obtain,
% \printcomment{\color{red}++++++++++ the following implication is not true. $0$ need not belong to $\d f_i(x_\Pi)$. We might have to assume differentiability ++++++++++ }
% \begin{align}
%     {\bf 0} \in \d \sum_{i \in S_1} f_i(x_{\Pi}) & = \sum_{i \in S_1} \d f_i(x_{\Pi}) = \d f_1(x_\Pi) + \sum_{i = 2}^{n} \d f_i(x_{\Pi}), \quad \text{ and }  \label{eqn:out_grad_1}\\
%     {\bf 0} \in \d \sum_{i \in S_2} f_i(x_{\Pi}) & = \sum_{i \in S_2} \d f_i(x_{\Pi}) = \sum_{i = 2}^{n} \d f_i(x_{\Pi}). \label{eqn:out_grad_2}
% \end{align}
% Substituting~\eqref{eqn:out_grad_2} in~\eqref{eqn:out_grad_1} we obtain,
\[\nabla g_i(x_\Pi) = {\bf 0}\]
% Repeating the above argument for sets $S_0$ and $S_1=S_0-\{i\}$, for each agent $i$, we obtain that
% \begin{align*}
%    {\bf 0} \in \d f_i(x_\Pi), \quad \forall i \in \{1, \ldots, \, n\}. 
% \end{align*}
As the cost functions are assumed to be convex, the above implies that,
\begin{align}
    x_\Pi \in \arg \min_x ~ g_i(x)
\end{align}
By repeating the above argument for each $i \in \{1, \ldots, \, n\}$, we have
\begin{align}
    x_\Pi \in \arg \min_x ~ g_i(x), \quad \forall i\in \{1, \ldots, \, n\}. \label{eqn:argmin_i}
\end{align}
Therefore, 
\begin{align}
    x_{\Pi}  \in \bigcap_{i = 1}^n \arg \min_x ~ g_i(x) \neq \emptyset.
    \label{e_nonempty}
\end{align}

Similarly, for every non-empty set of agents $T$
\begin{align}
    x_{\Pi}  \in \bigcap_{i\in T} \arg \min_x ~ g_i(x) \neq \emptyset.
    \label{e_nonempty_T}
\end{align}

Thus, $\bigcap_{i \in T} \arg \min ~ g_i(x)\neq\emptyset$.
Then, Lemma \ref{lem:non-empty-x} implies that
\begin{align}
\arg \min \sum_{i \in T} g_i(x) = \bigcap_{i \in T} \arg \min ~ g_i(x)~\neq \emptyset, \quad \forall \text{ non-empty } T \subseteq \{1, \ldots, \, n\}.
\label{eqn:equal_sets}
\end{align}

% \printcomment{++++++++++++++ doesn't above contradict \eqref{eqn:x_dag} already since $x^{\dagger}$ is cannot be a minimum? Why do we need to refer to $t$-resilience below to obtain a contradiction? +++++++++++++++++}
% As $\Pi$ is assumed to be {\em $t$-resilient}, $x_\Pi \in \arg \min \sum_{j \in S}f_j(x)$. Therefore, the above is contradiction of~\eqref{eqn:x_dag}. Therefore, $x^{\dagger} \in \bigcap_{j \in S} \arg \min ~ f_j(x)$. This, together with \eqref{e_nonempty}, implies that,
% \begin{align}
%     \arg \min \sum_{j \in S} f_j(x) \subseteq \bigcap_{j \in S} \arg \min ~ f_j(x) \neq \emptyset \label{e_s_1}
% \end{align}
% Recall that $S$ is any subset of $\{1, \ldots, \, n\}$ of size at least $n-t$.
% \printcomment{+++++++++ need to fix above depending on resolution of the previous comment ++++++++++ [\color{blue} We don't. I had the changes made in my personal copy. I have copied those above.]}

~
\noindent\hfil\rule{0.75\textwidth}{.4pt}\hfil
~\\

% As $S$ is an arbitrary subset of $\{1, \ldots, \, n\}$ of size at least $n-t$, the above implies,
% \begin{align}
%    \arg \min \sum_{i \in S} f_i(x) \subseteq \bigcap_{i \in S} \arg \min ~ f_i(x), \quad \forall S \subseteq \{1, \ldots, \, n\}, ~ \mnorm{S} \geq n-t.  \label{eqn:in_1}
%\end{align}

% It is trivial to show that 
% \begin{align}
% \bigcap_{i \in S} \arg \min ~ f_i(x) \subseteq \arg \min \sum_{i \in S} f_i(x), \quad \forall S \subseteq \{1, \ldots, \, n\}.
% \label{e_s_2}
% \end{align}
% \eqref{e_s_1} and \eqref{e_s_2} together imply that
% \begin{align}
%     \arg \min \sum_{i \in S} f_i(x) = \bigcap_{i \in S} \arg \min ~ f_i(x), \quad \forall S \subseteq \{1, \ldots, \, n\}, ~ \mnorm{S} \geq n-t.  \label{eqn:equal_sets}
% \end{align}

% \printcomment{+++++++++ so far we seem to have relied on only 1-resilience, not $t$-resilience ++++++++++[{\color{blue}Yes, but the condition derived above is weaker than 2t-redundancy, i.e., 2t-redundancy implies the above but the converse is not true. }] I added a footnote in the next sentence below to address this.}

% ~\\ \hrule ~\\

Now we consider execution $E_S$ (defined earlier) in which the nodes in set $S$ are non-faulty.
Using the results derived in the proof so
far,\footnote{Footnote \ref{footnote:k} noted that the notion of $2t$-redundancy can be extended to $k$-redundancy. The proof so far has relied only on 1-redundancy, which is weaker than $2t$-redundancy. The latter part of this proof makes use of $2t$-redundancy.}
we will show that, for any $\widehat{S}\subset S$ subject to $|\widehat{S}|\geq n-2t$, 
\[\arg \min_x \sum_{i \in S} g_i(x) = \bigcap_{i \in \widehat{S}} \arg \min_x ~ g_i(x).\]
The proof concludes once we have shown the above equality.\\ 

Consider an arbitrary subset $\widehat{S}\subset S$ subject to $\mnorm{\widehat{S}} =  n-2t$. It is trivially true that
\begin{align}
    \bigcap_{i \in S} \arg \min_x ~ g_i(x) \subseteq \bigcap_{i \in \widehat{S}} \arg \min_x ~ g_i(x). \label{eqn:subset_1}
\end{align}

So it remains to show that $\bigcap_{i \in S} \arg \min ~ g_i(x)$ is not a strict subset of $\bigcap_{i \in \widehat{S}} \arg \min ~ g_i(x)$.
The proof below is by contradiction.

~
\noindent\hfil\rule{0.75\textwidth}{.4pt}\hfil
~\\

Suppose that
\begin{align}
    \bigcap_{i \in S} \arg \min ~ g_i(x) \subset \bigcap_{i \in \widehat{S}} \arg \min ~ g_i(x). \label{eqn:not_true}
\end{align}
This implies that there exists a point 
\begin{align}
    x^{\dagger} \in \bigcap_{i \in \widehat{S}} \arg \min ~ g_i(x), \label{eqn:x_dag_s_hat}
\end{align}
such that
\begin{align}x^{\dagger} \not \in \bigcap_{i \in S } \arg \min_x ~ g_i(x).
     \label{eq:not-in-S}
\end{align}

Therefore, there exists an $i^{\dagger} \in S$ such that 
\begin{align}
     x^{\dagger} \not \in \arg \min ~ g_{i^\dagger}(x). \label{eqn:x_i_dag}
\end{align}
% Therefore, 
% \begin{align}
%    x^{\dagger} \not \in \bigcap_{i \in S} \arg \min ~ g_{i}(x). \label{eqn:x_i_dag_not_S}
% \end{align}

Let $C=\{1,\cdots,n\}\setminus S$ and
$F=S\setminus \widehat{S}$. Then $|C|=|F|=t$.
Now we define executions $E_C$ and $E_F$.
\begin{itemize}
\item {\em Execution $E_C$:} In execution $E_C$ the $t$ agents in set $C$ are faulty, and the $n-t$ agents in set $S$ are non-faulty.
In execution $E_C$, the behavior of each agent $i \in S$  is consistent with its true cost function being $g_i(x)$, which is identical to its true cost function in execution $E_S$. However, each faulty agent $j \in C$ behaves consistent with a differentiable and convex true cost function  $h_j(x)$ that has a unique minimum at $x^{\dagger}$.
\item
{\em Execution $E_F$:} In execution $E_F$ the $t$ agents in set $F$ are faulty, and the remaining $n-t$ agents in $\widehat{S} \cup C$ are non-faulty.
In execution $E_F$, the behavior of each agent $i \in S$ (including the faulty agents in $F$) is consistent with the cost function $g_i(x)$. 
% The fact that $i^\dagger \in T$ and~\eqref{eqn:equal_sets} together imply that,
% \begin{align}
%     x^{\dagger} \not \in \bigcap_{i \in T} \arg \min g_i(x)  = \arg \min \sum_{i \in T} g_i(x). \label{eqn:x_not_sum}
% \end{align}
Each non-faulty agent $j \in C$ behaves consistent with
it true cost function being $h_j(x)$, which is defined in execution $E_C$. Recall that each $h_j(x)$ has a unique minimum at $x^{\dagger}$.
\end{itemize}
Observe that the server cannot distinguish between executions $E_C$ and $E_F$.\\

Now, \eqref{eq:not-in-S} implies that $h_j(x)$ does not minimize at any point in $\bigcap_{i \in S} \arg \min_x ~ g_i(x)$. That is, for every agent $j \in C$, 
\begin{align}
\begin{split}
   & \{x^{\dagger}\} = \arg \min h_j(x), \text{ and } \\ \label{eqn:bias}
   & \left(\bigcap_{i \in S} \arg \min ~ g_i(x) \right) \bigcap  \arg \min h_j(x)  = \emptyset
\end{split}
\end{align}
%%%%%%%%%%%%%%%%% comment %%%%%%%%%%%%%%%%%%%%%%%%%%
\begin{comment}
We now define executions $E_T$ and $E_F$.
\begin{itemize}
\item
In execution $E_T$, suppose that each non-faulty agent $i$ in $T$ has cost function $f_i(x)$, and each faulty agent $j \in \C$ behaves
consistent with a convex cost function $h_j(x)$ has
a unique minimum at $x^{\dagger}$; thus, $h_j(x)$ does not minimize at any point in $\bigcap_{i \in T} \arg \min ~ f_i(x)$. That is, for every agent $j \in \C$, 
\begin{align}
\begin{split}
   & \{x^{\dagger}\} = \arg \min h_j(x), \text{ and } \\ \label{eqn:bias}
   & \left(\bigcap_{i \in T} \arg \min ~ f_i(x) \right) \bigcap  \arg \min h_j(x)  = \emptyset
\end{split}
\end{align}

\item In Execution $E_F$, the $n-t$ agents $\widehat{S}\cup C$ are non-faulty and the agents in $T\setminus (\widehat{S}\cup C)$
are Byzantine faulty. Each faulty agent $i$ behaves consistent with function $f_i(x)$. Each non-faulty agent $i\in\widehat{S}$
has cost function $f_i(x)$ and each non-faulty agent $i\in C$ has cost function $h_i(x)$ defined above.

\end{itemize}
The server cannot distinguish between the two executions above, and the $t$-resilient algorithm $\Pi$ must produce identical output in both cases.
\end{comment}
%%%%%%%%%%%%%%%%% comment %%%%%%%%%%%%%%%%%%%%%%%%%%
As $\Pi$ is {\em $t$-resilient}, in execution $E_F$, algorithm $\Pi$ must produce an output in
\begin{align}
\arg \min \left(\sum_{i\in \widehat{S}}g_i(x)+\sum_{j\in C}h_j(x)\right)
\label{e_f}
\end{align}
(Recall that the agents in $\widehat{S}\cup C$ are non-faulty in execution $E_F$.)\\

\eqref{eqn:x_dag_s_hat} and \eqref{eqn:bias} together imply that 
\[
\left(\bigcap_{i\in \widehat{S}} \arg\min g_i(x)\right) 
\bigcap
\left(\bigcap_{j\in C}h_j(x)\right) = \{x^\dagger\}
\]
That is, the above set contains only $x^\dagger$.
This, in turn, by Lemma \ref{lem:non-empty-x} implies that the set in \eqref{e_f} only
contains the point $x^\dagger$, and thus, algorithm $\Pi$
must output $x^\dagger$ in execution $E_F$.\\

Now, since algorithm $\Pi$ cannot distinguish between executions $E_F$ and $E_C$, it must also output $x^\dagger$ in execution
$E_C$ as well. However, from \eqref{eqn:equal_sets} and \eqref{eq:not-in-S}, respectively, we know that
\[
\bigcap_{i\in S}\arg\min g_i(x) = \arg\min \sum_{i\in S}g_i(x)
\]
and
\[
x^\dagger\not\in \bigcap_{i\in S}\arg\min g_i(x).
\]
% Thus,
% \[x^\dagger\not\in \arg\min \sum_{i\in T}f_i(x)
% \]
The above two equations imply that $x^\dagger\not\in\arg\min \sum_{i\in S}g_i(x)$, and $\Pi$ cannot output $x^\dagger$ in execution $E_C$ (otherwise $\Pi$ cannot be {\em $t$-resilient}).
This is a contradiction.\\

Therefore, we have proved that $\bigcap_{i \in S } \arg \min ~ g_i(x)$ is not a strict subset of $\bigcap_{i \in \widehat{S}} \arg \min ~ g_i(x)$.

~
\noindent\hfil\rule{0.75\textwidth}{.4pt}\hfil
~\\

Above result together with \eqref{eqn:subset_1} implies that 
\[\bigcap_{i \in S} \arg \min ~ g_i(x) = \bigcap_{i \in \widehat{S}} \arg \min ~ g_i(x).\]

Recall that $\widehat{S}$ is an arbitrary subset of $S$ with $\mnorm{\widehat{S}} = n-2t$. Therefore, the above implies that for every subset $\widehat{S}$ of $S$ with $\mnorm{\widehat{S}} \geq n-2t$,
\[\bigcap_{i \in S} \arg \min ~ g_i(x) = \bigcap_{i \in \widehat{S}} \arg \min ~ g_i(x).\]
This together with~\eqref{eqn:equal_sets} implies that
\[\arg \min \sum_{i \in S} g_i(x) = \bigcap_{i \in \widehat{S}} \arg \min ~ g_i(x), \quad \forall \text{ non-empty } \widehat{S} \subseteq S, ~ \mnorm{\widehat{S}} \geq n-2t.\]
Thus, if $\Pi$ is {\em $t$-resilient} then the true cost functions of the agents satisfy the {\em $2t$-redundancy} property as stated in Definition~\ref{def:2t_red}. Hence, proving the necessity of {\em $2t$-redundancy} property for {\em $t$-resilience}.
\end{proof}
~

The following collaborative optimization algorithm proves the sufficiency of {\em $2t$-redundancy} for {\em $t$-resilience}.

\subsection{A {\em $t$-resilient} algorithm}
\label{sub:t_algo}

We present an algorithm and prove that it is {\em $t$-resilient} if the agents satisfy the {\em $2t$-redundancy} property stated in Definition~\ref{def:2t_red_alt} or~\ref{def:2t_red}. We will suppose that Assumption~\ref{asp:basic} holds true and $n > 2t$. We only consider the case when $t > 0$, since the case of $t=0$ is trivial.\\

% there are some faulty agents in the system. Otherwise, the collaborative optimization problem can be solved trivially using the existing algorithms~\cite{boyd2011distributed, nedic2009distributed}.\\

% to be $\{h_1, \ldots, ~ h_n\}$. If agent $i$ is non-faulty then $h_i = f_i$, otherwise $h_i$ is an arbitrary function. \\

{\bf $t$-Resilient Algorithm:}
The server collects full description of the cost function of each agent. Suppose
that the server obtains cost function $h_j(x)$ from each agent $j\in\{1,\cdots,n\}$.
For each non-faulty agent $i$, $h_i(x)$ is the agent's true objective function.\\

The proposed algorithm outputs a point $x^*$ such that there exists a set $A$ of $n-t$ agents such that
for any $\widehat{A}\subset A$ with $|\widehat{A}|=n-2t$,
\[ x^*\in\arg\min \sum_{i\in \widehat{A}} h_i(x)
\]
If there are multiple candidate points that satisfy the condition above, then any one such point is chosen as the output.\\

Now we prove the correctness of the above algorithm if $2t$-redundancy holds.

\begin{proof}
Assume that the $2t$-redundancy property holds.
First we observe that the algorithm will always be able to output a point
if $2t$-redundancy is satisfied.
Let $S$ denote the set of all non-faulty agents. Recall that $\mnorm{S} \geq n-t$.
In particular, consider a set $A$ that consists of any $n-t$ non-faulty agents, that is, $A\subseteq S$. 
For any $\widehat{A}\subset A$
where $|\widehat{A}|=n-2t$, due to $2t$-redundancy (Definition \ref{def:2t_red}) and Assumption~\ref{asp:basic}, we have
\begin{align}\bigcap_{i \in \widehat{A}} ~ \arg \min_{x \in \R^d} h_i(x) = \arg \min_{x \in \R^d} ~ \sum_{i \in S}  h_i(x) \end{align}
This implies that every point in $\arg\min \sum_{i \in S}  h_i(x)$
is a candidate for the output of the algorithm. Additionally,
due to Assumption \ref{asp:basic},
$\arg\min_x \sum_{i \in S}  h_i(x)$ is guaranteed to be non-empty.
Thus, the algorithm will always produce an output.\\

Next we show that the algorithm achieves $t$-resilience.
Consider any set $A$ for which the condition in the algorithm is true.
The algorithm outputs $x^*$.
From the algorithm, we know that
for any $\widehat{A}\subset A$ with $|\widehat{A}|=n-2t$,
\[ x^*\in\arg\min \sum_{i\in \widehat{A}} h_i(x)
\]
Now, since at most $t$ agents are faulty, there exists at least one set $\widehat{S}$ containing
$n-2t$ non-faulty agents such that $\widehat{S}\subseteq A$ (and also $\widehat{S}\subseteq S$).
Thus, 
\begin{align} x^*\in\arg\min \sum_{i\in \widehat{S}} h_i(x)
\label{e_0000}
\end{align}
Also, since $\widehat{S}\subseteq S$, due to $2t$-redundancy (Definition~\ref{def:2t_red}),  we have
\begin{align}\bigcap_{i \in \widehat{S}} \arg \min h_i(x) = \arg \min \sum_{i \in S}  h_i(x) \label{e_1111}
 \end{align}
Since $\arg \min \sum_{i \in S} h_i(x)$ is non-empty, the last equality implies that $\bigcap_{i \in \widehat{S}} \arg \min h_i(x)$ is non-empty. This, in turn, by Lemma \ref{lem:non-empty-x} implies that
\[ \arg\min \sum_{i\in \widehat{S}} h_i(x) = \bigcap_{i \in \widehat{S}} \arg \min h_i(x)\]
The last equality, \eqref{e_0000} and \eqref{e_1111} together imply that 
$$x^*\in \arg \min \sum_{i \in S}  h_i(x).$$
Thus, the above algorithm achieves $t$-resilience.
\end{proof}
~

It should be noted that the correctness of the $t$-resilient algorithm presented above does not require differentiability or convexity of the agents' true cost functions. Therefore, the $2t$-redundancy is a sufficient condition for $t$-resilience even when the agents' cost functions are non-differentiable and non-convex.\\

\noindent {\bf Alternate $t$-resilient algorithms:} There exist other, and more practical, algorithms to achieve \mbox{$t$-resilience} when $2t$-redundancy holds. However, there is a trade-off between algorithm complexity and additional properties assumed for the cost functions.
\begin{itemize}
    \item We present an alternate, computationally simpler, {\em $t$-resilient} algorithm in Section~\ref{sec:algo} for the case when the minimum values of each true cost function is zero.
    \item In our prior work~\cite{gupta2019byzantine}, we proposed a gradient-descent based distributed algorithm that is {\em $t$-resilient} if the cost functions have certain additional properties presented in Section~\ref{sec:norm}. The algorithm uses a computationally simple ``comparative gradient clipping'' mechanism to tolerate Byzantine faults.
\end{itemize}

\subsection{Prior work on redundancy}
\label{sub:red_prior}

To the best of our knowledge, there is no prior work on the tightness of {\em $2t$-redundancy} property for {\em $t$-resilience} in collaborative optimization. Nevertheless, it is worthwhile to note that conditions with some similarity to $2t$-redundancy are known to be necessary and sufficient for fault-tolerance in other systems, such as information coding and collaborative multi-sensing (or {\em sensor fusion}), discussed below. We note that collaborative multi-sensing can be viewed as a special case of the collaborative optimization problem presented in this report.\\

{\bf Redundancy for error-correction coding:} Digital machines store or communicate information using a finite length sequence of {\em symbols}. However, these {\em symbols} are may become erroneous due to faults in the system or during communication. A way to recover the information despite such error is to use an error-correction {\em code}. An error-correction code transforms (or {\em encodes}) the original sequence of symbols into another sequence of symbols called a {\em codeword}. It is well-known that a {\em code} that generates codewords of length $n$ can correct (or tolerate) up to $t$ symbols errors if and only if the Hamming distance between any two codewords of the {\em code} is at least $2t + 1$~\cite{lindell2010introduction, van1971coding}. There exist codes (e.g., Reed-Solomon codes) such that the sequence of symbols encoded in a codeword  can be uniquely determined using any $n-2t$ correct symbols of the codeword. \\

{\bf Redundancy for fault-tolerant state estimation:} The problem of collaborative optimization finds direct application in distributed sensing~\cite{rabbat2004distributed}. In this problem, the system comprises multiple sensors, and each sensor makes partial observations about the state of the system. The goal of the sensors is to collectively compute the complete state of the system. However, if a sensor is faulty then it may share incorrect observations. The problem of fault-tolerance in collaborative sensing for the special case wherein the sensors' observations are {\em linear} in the system state has gained significant attention in recent years~\cite{bhatia2015robust, chong2015observability, pajic2017attack, pajic2014robustness, shoukry2015imhotep, shoukry2017secure, su2018finite}. Chong et al., 2015~\cite{chong2015observability} and Pajic et al., 2015~\cite{pajic2014robustness} showed that the system state can be accurately computed when up to $t$ (out of $n$) sensors are faulty {\em if and only if} the system is {\em $2t$-sparse observable}, i.e., the state can be computed uniquely using observations of only $n-2t$ non-faulty sensors. We note that the property of {\em $2t$-sparse observability} is a special instance of the more general {\em $2t$-redundancy} property presented in this report. Moreover, the necessity and sufficiency of the \mbox{\em $2t$-redundancy} property proved in this report implies the necessity and sufficiency of {\em $2t$-sparse observability} for fault-tolerant state estimation for a more general setting wherein the sensor observations may be non-linear; however, the converse is not true.\\

Next, we consider the case when the cost functions are independent, and may not satisfy the {\em $2t$-redundancy} property. 

%%%%%%%%%%%%%%%%%%%%%%%%%%%%% LOWER BOUND %%%%%%%%%%%%%%%%%%%%%%%%%%%%%%%%
% \input{lower_bound}
% \input{independent_cost_v2}

%%%%%%%%%%%%%%%%%%%%%%%%%%%%%%%%% PROPOSED ALGORITHM %%%%%%%%%%%%%%%%%%%%%%%%%%%%%%%%%%%%%%%
\section{The case of Independent Cost Functions}
\label{sec:ind}

In this section, we present the case when the true cost functions of the agents are independent. Throughout this section we assume that $t > 0$, otherwise the problem of resilience is trivial.\\

We show below by construction that when the true cost functions are non-negative then {\em $(u, t)$-weak resilience} can be achieved for $u \geq t$ even if the true cost functions are independent. Note that, by Definition~\ref{def:weak_res}, when the true cost functions of the agents are non-negative then {\em $(u, t)$-weak resilience} trivially implies {\em $(u^\dagger, t)$-weak resilience} where $u^\dagger \geq u$. Therefore, achievability of {\em $(t, \, t)$-weak resilience} implies the achievability of {\em $(u, \, t)$-weak resilience} for all $u \geq t$.\\

In the subsequent subsection we present a collaborative optimization algorithm that guarantees {\em$(t, \, t)$-weak resilience} when the true cost functions are non-negative and $n > 2t$. 
In Section \ref{ss:t}, we show that the algorithm below also achieves $t$-resilience
under certain conditions.

% \printcomment{++++++ added last sentence above, and created section 3.2 +++++++}

\subsection{Algorithm for {\em $(t,\,t)$-Weak Resilience}}
\label{sec:algo}

In the proposed algorithm, the server obtains a full
description of the agents' cost functions.
We denote the function obtained by the server from agent $i$ as $h_i(x)$.
Let the true cost function of each agent $i$ be denoted $f_i(x)$. Then for each non-faulty agent $i$,
$h_i(x) = f_i(x), ~ \forall \, x$. On the other hand, for each faulty agent $i$, $h_i(x)$ may not necessarily equal $f_i(x)$.
\\

% To present our algorithm, we denote the objective functions shared by the agents by $h_1, \ldots, \, h_n$. Recall that if agent $i$ is corrupted then $h_i$ may be an arbitrary function, if agent $i$ is not corrupted then $h_i = f_i$. Throughout this section, we will assume that Assumption~\ref{asp:min} holds, i.e., the objective functions of the non-faulty agents are non-negative. \\
% If any of the functions in set of {\em effective} objective functions $\{h_1, \ldots, \, h_n\}$ violate Assumption~\ref{asp:min}, it will be deemed corrupted and will be eliminated.  \\

The algorithm comprises three steps: 

\begin{itemize}
    \item\textit{Pre-processing Step}: For any agent $j$, if $h_j(x)$ is not non-negative for some $x$ or $\min h_j(x)$ is not finite (or does not exist), then $j$ must be faulty. Remove $j$ from the system. Decrement $t$ and $n$ each by 1 for each agent thus removed. In other words,
    the cost functions of the remaining agents are non-negative. Also, it is easy to see that if the faulty agents are in the minority then $n>2t$ after pre-processing for the updated values of $n$ and $t$.\footnote{A worst-case adversary may ensure that $h_i(x)$ for faulty agent $i$ is non-negative, so that no faulty agents will be eliminated in the pre-processing step.}

    \item\textit{Step 1}: For each set $A$ of agents such that $|A|=n-t$, compute 
    \[v_A = \min_{x \in \R^d} \, \sum_{i \in A} h_i(x).\]
    \item\textit{Step 2}: Determine a subset $\widehat{A}$ of size $n-t$ such that 
    \begin{align}
        v_{\widehat{A}} = \min \left\{ v_A ~ \vline ~ A \subseteq \{1, \ldots, \, n\}, ~ \mnorm{A} = n-t \right\} \label{eqn:vs_hat}
    \end{align}
    % \[v_{\widehat{\S}} = \min \left\{ v_S ~ \vline ~ S \subseteq \{1, \ldots, \, n\}, ~ \mnorm{S} = n-t \right\} \l.\]
{\bf Output} a point $\widehat{x} \in \arg \min_x  \sum_{i \in \widehat{A}} h_i(x)$.
\end{itemize}

~
\noindent\hfil\rule{0.75\textwidth}{.4pt}\hfil
~\\

% ~
% \hrule
% ~\\

Now we prove that the algorithm is {\em $(t,t)$-weak resilient}. It should be noted that the {\em $(t,t)$-weak resilience} property of the algorithm holds true despite the true cost function being {\em non-convex} and {\em non-differentiable}.

\begin{theorem}
\label{thm:algo}
Suppose that Assumption~\ref{asp:basic} holds, and $n> 2t$. If the true cost functions are non-negative then the above algorithm is $(t, \, t)$-weak resilient.
\end{theorem}
\begin{proof}
Suppose that, before the pre-processing step $n-t = a$ and $n-2t = b$. In the proof, we consider the set of agents, and the values of $n$ and $t$ after the pre-processing step of the algorithm. 
In the worst-case for the algorithm, all faulty agents will send non-negative functions, thus, no faulty agents are removed in the pre-processing step. Also observe that, in general, for the updated values of $n$ and $t$ after the pre-processing step,
 (i) $n-t = a$ (i.e., $n-t$ remains unchanged), and (ii) $n-2t \geq b$, and (iii) $n>2t$.\\

For an execution of the proposed algorithm, let $\F$ denote the set of up to $t$ faulty agents, and let $S$ denote the set of non-faulty agents. Thus,
$|S|+|\F|=n$.\\

Recall the definition of $\widehat{A}$ in the algorithm above.
Let
\begin{align}
    S_1 = S \cap \widehat{A}\\
    F_1 = \F \cap \widehat{A}
\end{align}
Thus, $\widehat{A}=S_1\cup F_1$.
Since $\mnorm{\widehat{A}}=n-t$ and $\mnorm{\F}\leq t<n/2$,
we have that $\mnorm{S_1}\geq \mnorm{S} - t$ and $\mnorm{F_1}\leq t$.\\

% As $\mnorm{\C} \leq t < n/2$, and the size of $\widehat{\S}$ is $n-t$, there exists a non-empty set $\S_1 \subseteq \C^-$ of size at least $n - t - \mnorm{\C}$, and a set $\S_2 \subseteq \C$ of size at most $\mnorm{\C}$ such that 
% \begin{align}
%    \widehat{\S} = \S_1 \cup \S_2. \label{eqn:decomp_S}
% \end{align}

First, note that owing to the {\em pre-processing step} and Assumption~\ref{asp:basic}, for every set of $n-t$ agents $A$, $v_{A} = \min \, \sum_{i \in A} h_i(x)$ exists and is finite. \\

Now, note that
\[v_{\widehat{A}} = \min \sum_{i \in \widehat{A}} h_i(x) = \sum_{i \in \widehat{A}} h_i(\widehat{x}) = \sum_{i \in  S_1} h_i(\widehat{x}) + \sum_{j \in F_1} h_j(\widehat{x}).\]
From~\eqref{eqn:vs_hat}, $v_{\widehat{A}}  \leq v_A$ for all sets $A$ of size $n-t$. Therefore, there exists a subset $S' \subseteq S$ with $\mnorm{S'}=n-t$ such that 
\[v_{\widehat{A}} \leq v_{S'}.\]
From above we obtain,
\begin{align*}
    \sum_{i \in  S_1} h_i(\widehat{x}) + \sum_{j \in F_1} h_j(\widehat{x}) \leq v_{S'} = 
    \min \sum_{i \in S'} h_i(x) . 
\end{align*}
Recall that $h_i(x)= f_i(x)$ for all $i \in S$. As $S_1$ and $S'$ are subsets of $S$, the above implies that,
\begin{align}
    \sum_{i \in  S_1} f_i(\widehat{x}) + \sum_{j \in F_1} h_j(\widehat{x}) \leq v_{S'} = \min \sum_{i \in S'} f_i(x) . \label{eqn:pen_1}
\end{align}
Each $h_j(x)$ is a non-negative function (due to the pre-processing step).
Therefore, $h_j(\widehat{x}) \geq 0$ for all $j \in F_1$. Substituting this in~\eqref{eqn:pen_1} implies,
\begin{align}
    \sum_{i \in  S_1} f_i(\widehat{x}) \leq \min \sum_{i \in S'} f_i(x). \label{eqn:pen}
\end{align}
As $S'  \subseteq S$, non-negativity of cost functions implies that, 
\[\min \sum_{i \in S'} f_i(x) \leq \min \sum_{i \in S} f_i(x).\]
Substituting the above in~\eqref{eqn:pen} implies,
\begin{align}
    \sum_{i \in  S_1} f_i(\widehat{x}) \leq \min \sum_{i \in S} f_i(x). \label{eqn:proved}
\end{align}
Recall that $\mnorm{S_1} \geq \mnorm{S} - t$. Recalling that the set of non-faulty agents is not affected by the pre-processing step, the above implies that the proposed algorithm is $(t, \,t)$-weak resilient.
% Repeating the above argument for all possible
% sets of faults $\F$, \eqref{eqn:proved} implies that the proposed algorithm is $(t, \,t)$-weak resilient.
% \printcomment{+++++ Do we need to repeat the argument for all possible F, as stated above? ++++++++++++ {\color{blue} No. Removed.}}
\end{proof}
~

% \printcommentNiru{ ++++ New Subsection Added +++++}
The algorithm above is {\em $(t,\, t)$-weak resilient} for the case when each true cost function is non-negative. However, in general, there may exist collaborative optimization algorithms that are {\em $(t,\, t)$-weak resilient} only for the case when each true cost function has minimum value $0$. We present below a normalization technique for generalizing the weak resilience of such algorithms. Specifically, given a collaborative optimization algorithm that is {\em $(u,\, t)$-weak resilient} for the case when each true cost function has minimum value $0$, the presented normalization technique generalizes the algorithm to the case when the true cost
functions are non-negative.\\

% Next, we present a simple normalization technique that renders a collaborative optimization algorithm, that is {\em $(u, t)$-weak resilient} in the special case when the minimum values of the true cost functions are zero, {\em $(u, t)$-weak resilient} in the more general case when the true cost functions are non-negative. 
Later, we will see that the normalization technique renders a collaborative optimization algorithm that is {\em $(t, t)$-weak resilient} for the case when the true cost functions are non-negative, such as the one presented above, {\em $t$-resilient} if the true cost functions satisfy the {\em $2t$-redundancy} property.

\subsection{Normalized Implementation of {$(u, \, t)$-weak resilient} algorithm}
\label{sub:normalize}

We denote the function obtained by the server from agent $i$ as $h_i(x)$. Let the true cost function of each agent $i$ be denoted $f_i(x)$. Then for each non-faulty agent $i$,
$h_i(x) = f_i(x), ~ \forall \, x$. For each faulty agent $i$, $h_i(x)$ may not necessarily equal $f_i(x)$. \\

Consider an arbitrary algorithm $\Pi$ that achieves \mbox{\em $(u, \, t)$-weak} {\em resilience} when each true cost function has minimum value $0$. With $\Pi$ as a building block, we design an algorithm
$\Pi^+$ using the two-step normalization procedure below. We will refer to $\Pi^+$ as the {\em normalized implementation} of $\Pi$. 

% The \textbf{\em normalized implementation} of algorithm $\Pi$, which we denote by $\Pi^+$, comprises two steps.
% We use the two-step normalization procedure below to design an algorithm
% $\Pi^+$ that uses $\Pi$ as a building block.

\begin{itemize}
    \item\textit{Step 1}: For each agent $i$, compute $\min h_i(x)$. 
    If $\min h_i(x)$ does not exist or is infinite then remove agent $i$ from the system.
    % \footnote{It is easy to see that no non-faulty agent is removed when Assumption~\ref{asp:basic} holds true.} 
    Decrement $n$ and $t$ each by 1 for each agent thus removed. Otherwise, define an alternate {\em effective cost function} $h^{\dagger}_i$ such that
    \begin{align}
        h^{\dagger}_i(x) = h_i(x) - \min_{x \in \R^d} h_i(x), \quad \forall x \in \R^d. \label{eqn:nm_cost}
    \end{align}
    It is easy to see that if the faulty agents are in the minority (i.e., $n > 2t$ prior to the normalization step) then $n>2t$ upon completion of the normalization step for the updated values of $n$ and $t$.\footnote{A worst-case adversary may ensure that $h_i(x)$ for faulty agent $i$ is non-negative, so that no faulty agents will be eliminated in the normalization step.}\\
    
    The agents that remain after the above step are numbered $1$ through $n$, without loss of generality.
    
    \item\textit{Step 2}: Execute $\Pi$ on the {\em effective cost functions} $h^{\dagger}_1(x), \cdots, ~ h^{\dagger}_n(x)$.
\end{itemize}

% We will refer to $\Pi^+$ as the {\em normalized implementation} of algorithm $\Pi$

The resilience property of algorithm $\Pi^{+}$ is stated below.

\begin{lemma}
\label{lem:zero_non-zero}
Suppose that Assumption~\ref{asp:basic} holds true. If algorithm $\Pi$ is {\em $(u, t)$-weak resilient} when the true cost function of each agent has minimum value equal to zero then $\Pi^+$, the {\em normalized implementation} of $\Pi$, is {\em $(u, t)$-weak resilient} when the true cost functions are non-negative.
\end{lemma}
\begin{proof}
In the proof, we consider the set of agents, and the values of $n$ and $t$ after the {\em step 1} of the normalization procedure. Note that, due to Assumption~\ref{asp:basic}, the set of non-faulty agents is not affected by step 1. Let the true cost function of each agent $i$ be denoted $f_i$. The true cost 
functions are assumed non-negative.\\
% Also, in the worst-case for the algorithm, all faulty agents will send functions that have finite minimum values, and thus, no faulty agents are removed in the step 1. Let the true cost function of each agent $i$ be denoted $f_i$. The true cost functions are assumed non-negative.\\

Suppose that algorithm $\Pi$ is $(u, t)$-weak resilient when each true cost function has minimum value equals zero. For an execution of the algorithm $\Pi^+$, let set $S$ with $\mnorm{S} \geq n-t$ denote the set of non-faulty agents. Let the output of $\Pi^+$ be denoted by $\widehat{x}$. \\

Due to Assumption~\ref{asp:basic}, for each agent $i$, $\min_{x \in \R^d} f_i(x)$ exists and is finite. For each agent $i$, let $f^{\dagger}_i$ denote a function such that
\begin{align}
    f^{\dagger}_i(x) = f_i(x) - \min_{y \in \R^d} f_i(y), \quad \forall x \in \R^d. \label{eqn:effective_true_cost}
\end{align}
Therefore, for each agent $i$, 
\begin{align}
    \min_{x \in \R^d} f^{\dagger}_i(x) = \min_{x \in \R^d} \left( f_i(x) - \min_{y \in \R^d} f_i(y) \right) = \min_{x \in \R^d} f_i(x) - \min_{y \in \R^d} f_i(y) = 0. \label{eqn:zero_min_non-faulty}
\end{align}
% In the rest of the proof, we will write notation `$\min_{x \in \R^d} (\cdot)$' simply as `$\min (\cdot)$' unless otherwise mentioned. \\

Note that from~\eqref{eqn:nm_cost} in the {\em step 1}, if an agent $i$ is non-faulty then $h^{\dagger}_i = f^{\dagger}_i$. Therefore, the true cost functions in {\em step 2}, i.e., during the execution of algorithm $\Pi$, are $f^\dagger_1, \cdots, ~ f^\dagger_n$. This together with~\eqref{eqn:zero_min_non-faulty} implies that each true cost function has minimum value equal to $0$ during the execution of $\Pi$. As $\Pi$ is assumed $(u, t)$-weak resilient for the case when each true cost function has minimum value equal to zero, by Definition~\ref{def:weak_res}, there exists $\widehat{S} \subseteq S$ with $\mnorm{\widehat{S}} \geq \mnorm{S} - u$ such that
\begin{align*}
    \sum_{i \in \widehat{S}} f^\dagger_i(\widehat{x}) ~ \leq ~ \min_{x \in \R^d} \, \sum_{i \in S} f^\dagger_i(x).
\end{align*}
Substituting from~\eqref{eqn:effective_true_cost} above we obtain,
\begin{align*}
    \sum_{i \in \widehat{S}} \left( f_i(\widehat{x}) - \min_{y \in \R^d} f_i(y) \right) ~ \leq ~ \min_{x \in \R^d} \, \sum_{i \in S} \left( f_i(x) - \min_{y \in \R^d} f_i(y)\right).
\end{align*}
Trivially, for each $i$, 
\[\min_{x \in \R^d} \left( \min_{y \in \R^d} f_i(y) \right) = \min_{y \in \R^d} f_i(y).\]
Therefore, from above we obtain,
\begin{align*}
    \sum_{i \in \widehat{S}} f_i(\widehat{x}) -  \sum_{i \in \widehat{S}} \min_{y \in \R^d} f_i(y) ~ \leq ~ \min_{x \in \R^d} \, \sum_{i \in S} f_i(x) - \sum_{i \in S} \, \min_{y \in \R^d} f_i(y).
\end{align*}
Upon rearranging the terms we obtain,
\begin{align}
    \sum_{i \in \widehat{S}} f_i(\widehat{x}) ~ \leq ~ \min_{x \in \R^d} \, \sum_{i \in S} f_i(x) - \sum_{i \in S \setminus \widehat{S}} \, \min_{y \in \R^d} f_i(y). \label{eqn:out_norm_pi_1}
\end{align}
As the true cost functions $f_1, \ldots, ~ f_n$ are assumed non-negative, i.e., $f_i(x) \geq 0$ for all $x$ and $i$, then $\min_{y \in \R^d} f_i(y) \geq 0$ for all $i$. From substituting this in~\eqref{eqn:out_norm_pi_1} we obtain,
\begin{align*}
    \sum_{i \in \widehat{S}} f_i(\widehat{x}) ~ \leq ~ \min_{x \in \R^d} \, \sum_{i \in S} f_i(x).
\end{align*}
Thus, by Definition~\ref{def:weak_res}, the {\em normalize implementation} of algorithm $\Pi$, i.e., $\Pi^+$, is {\em $(u, \, t)$-weak resilient} when the true cost functions are non-negative.
\end{proof}

\subsection{$t$-Resilience Property}
\label{ss:t}
% A formal proof of Theorem~\ref{thm:algo} is deferred to Appendix~\ref{app:proof_algo}. 
In this section, we show that if the true cost functions are non-negative, and satisfy the {\em $2t$-redundancy} property, then the normalized implementation of a {\em $(t, \,t)$-weak resilient} algorithm, such as the one presented above, is also {\em $t$-resilient}. First, let us consider the special case wherein each true cost function has minimum value equal to zero.
% show the  if the true cost functions of the agents satisfy the {\em $2t$-redundancy} property, and have minimum values equal to zero, then a \mbox{\em $(t, \,t)$-weak} {\em resilient} algorithm is \mbox{\em $t$-resilient}, and equivalently, \mbox{\em $(0, \,t)$-weak resilient}. 

\begin{lemma}
\label{lem:weak_strong}
Suppose that Assumption~\ref{asp:basic} holds true, and $n > 2t$. If the true cost functions of the agents satisfy the {\em $2t$-redundancy} property, and each true cost function has minimum value equal to zero, then a \mbox{\em $(t, \,t)$-weak} {\em resilient} algorithm is also \mbox{\em $t$-resilient}.
\end{lemma}
\begin{proof}
Let $\Pi$ be a \mbox{\em $(t, \,t)$-weak resilient} collaborative optimization algorithm. Consider an execution of $\Pi$, named $E_\F$, where $\F$ denotes the set of faulty agents with $\mnorm{\F} \leq t$. The remaining agents in $S = \{1, \ldots, \, n\} \setminus \F$ are non-faulty. Suppose that the true cost function of each agent $i$ in execution $E_\F$ is $f_i$.\\

As $E_\F$ is an arbitrary execution, to prove the lemma it suffices to show that the output of $\Pi$ in execution $E_\F$ is a minimum of the sum of the true cost functions of all the non-faulty agents $S$.\\

We have assumed that the minimum values of the functions $f_1(x), \ldots, \, f_n(x)$ are zero, i.e., 
\begin{align}
    \min_{x \in \R^d} f_i(x) = 0, \quad  ~ 1 \leq i \leq n. \label{eqn:min_zero}
\end{align}
In the rest of the proof, the notation `$\min_{x \in \R^d}$' is simply written as `$\min$' unless otherwise noted.\\

By applying the condition in Definition~\ref{def:2t_red} of {\em $2t$-redundancy} property for all
possible $\widehat{S}\subseteq S$ (where $|\widehat{S}| \geq n-2t$) we can conclude that the set
$\arg\min \sum_{i \in S} f_i(x)$ is contained in  the set $\arg\min f_i(x)$ for each $i\in S$. 
This, and the fact that each individual cost function has minimum value 0, implies that
\begin{align*}
    \min \sum_{i \in S} f_i(x) = \sum_{i \in S} \min f_i(x).
\end{align*}
Substituting from~\eqref{eqn:min_zero} above implies that
\begin{align}
    \min \sum_{i \in S} f_i(x) = 0.  \label{eqn:v_0}
\end{align}
Let $x_\Pi$ denote the output of $\Pi$. As $\Pi$ is \mbox{\em $(t, \,t)$-weak resilient}, there exists a subset $\widehat{S}$ of $S$ of size $\mnorm{S} - t$ such that
\[\sum_{i \in \widehat{S}} f_i(x_\Pi) \leq \min \sum_{i \in S} f_i(x).\]
Substituting from~\eqref{eqn:v_0} above implies that 
\[\sum_{i \in \widehat{S}} f_i(x_\Pi) \leq 0.\]
From~\eqref{eqn:min_zero}, $f_i(x_\Pi) \geq 0, ~ \forall  \, i$. The above implies that 
\begin{align*}
    f_i(x_\Pi) = 0, \quad \forall \, i \in \widehat{S}.
\end{align*}
Alternately,
\begin{align}
    x_\Pi \in \bigcap_{i \in \widehat{S}} \arg \min f_i(x). \label{eqn:x_int_S}
\end{align}
% \printcomment{+++++++++++ the first = sign below was $\geq$: I changed it to =. Verify that this is OK. ++++++++++{\color{blue} YES}}
As $\mnorm{\widehat{S}} = \mnorm{S} - t \geq n - 2t$, the {\em $2t$-redundancy} property implies that
\[\bigcap_{i \in \widehat{S}} \arg \min f_i(x) = \arg \min \sum_{i \in S} f_i(x).\]
From substituting the above in~\eqref{eqn:x_int_S} we obtain,
\[x_\Pi \in \arg \min \sum_{i \in S} f_i(x).\]
Thus, algorithm $\Pi$ achieves $t$-resilience.
\end{proof}
~

% The {\em $2t$-redundancy} property is preserved by the {\em normalization step} in the {\em normalized implementation} of $\Pi$. 
Utilizing the Lemma~\ref{lem:weak_strong} we show that the {\em normalized implementation} of a {\em $(t, \,t)$-weak resilience} is {\em $t$-resilience} when the true cost functions are non-negative, and satisfy the {\em $2t$-redundancy} property. Specifically, we have the following theorem. 

% \printcommentNiru{ ++++ editing ++++}

\begin{theorem}
\label{thm:weak_strong}
Suppose that Assumption~\ref{asp:basic} holds true, $n > 2t$, and we are given an algorithm $\Pi$ that is \mbox{\em $(t, \,t)$-weak} {\em resilient} when each true cost function has minimum value $0$. Then the algorithm $\Pi^+$ obtained as the {\em normalized implementation} of $\Pi$ is \mbox{\em $t$-resilient} when the true cost functions are non-negative, and satisfy the {\em $2t$-redundancy} property. 
\end{theorem}
\begin{proof}
Let the true cost functions of each agent $i$ be denoted by $f_i$. The true cost functions are assumed to be non-negative, i.e,
\begin{align*}
    f_i(x) \geq 0, \quad \forall x \in \R^d, ~ i \in \{1, \ldots, \, n\}.
\end{align*}
The true cost functions $f_1, \ldots, \, f_n$ are also assumed to satisfy the {\em $2t$-redundancy} property, i.e., the condition stated in Definition~\ref{def:2t_red_alt} holds true.\\

Suppose that algorithm $\Pi$ is a \mbox{\em $(t, \,t)$-weak} {\em resilient}. Consider the {\em normalized implementation} of algorithm $\Pi$ presented in Section~\ref{sub:normalize}. It is easy to see that if $n > 2t$ a priori then $n > 2t$ upon completion of the {\em step 1} for the updated values of $n$ and $t$. Also, due to Assumption~\ref{asp:basic}, the set of non-faulty agents are not affected by the {\em step 1}. For the rest of the proof, we consider the set of agents, and the values of $n$ and $t$ after the {\em step 1}.\footnote{In the worst-case for the algorithm, all faulty agents will send functions that have finite minimum values, thus, no faulty agents are removed in the execution step.} 

Recall that due to Assumption~\ref{asp:basic}, for each $i$, $\min_{y \in \R^d} f_i(y)$ exists and finite. Note that, due to~\eqref{eqn:nm_cost} in {\em step 1}, the true cost function of each agent $i$ during the execution of $\Pi$ in Step 2, denoted by $f^\dagger_i$, satisfies the following:
\begin{align}
    f^\dagger_i(x) = f_i(x) - \min_{y \in \R^d} f_i(y), \quad \forall x \in \R^d. \label{eqn:2_f_dag}
\end{align}
As $$\min_{x \in \R^d} \left(\min_{y \in \R^d} f_i(y)\right) = \min_{y \in \R^d} f_i(y),$$ for each $i$, $\min_{x \in \R^d} f^\dagger_i(x) = 0$ and 
\begin{align}
    \arg \min_{x \in \R^d} f^\dagger_i(x) =  \arg \min_{x \in \R^d} \left( f_i(x) - \min_{y \in \R^d}f_i(y) \right) = \arg \min_{x \in \R^d} f_i(x). \label{eqn:preserve_argmin}
\end{align}
Now, consider two arbitrary sets of agents $S_1$ and $S_2$ each of size $n-2t$. As the true cost functions $f_i$'s are assumed to satisfy the {\em $2t$-redundancy} property, by Definition~\ref{def:2t_red_alt} and Assumption~\ref{asp:basic}, 
\[ \emptyset \neq \bigcap_{i \in S_1} \arg \min_{x \in \R^d} f_i(x) = \bigcap_{i \in S_2} \arg \min_{x \in \R^d} f_i(x).\]
Substituting from~\eqref{eqn:preserve_argmin} above we obtain, 
\[ \emptyset \neq \bigcap_{i \in S_1} \arg \min_{x \in \R^d} f^\dagger_i(x) = \bigcap_{i \in S_2} \arg \min_{x \in \R^d} f^\dagger_i(x).\]
As the above holds for any two such subsets $S_1$ and $S_2$, the cost functions $f^\dagger_1, \ldots, \, f^\dagger_n$ satisfy the {\em $2t$-redundancy} property.\\

The above together with Lemma~\ref{lem:weak_strong} implies that $\Pi$, which is executed in the {\em step 2}, is \mbox{\em $t$-resilient} when the true cost function of each agent $i$ is $f^\dagger_i$. Now, consider an execution of $\Pi^+$, the algorithm obtained as the {\em normalized implementation} of $\Pi$, where $S$ denotes the set of non-faulty agents. Let, $\widehat{x}$ denote the output of this execution. Then,
\begin{align}
    \widehat{x} \in \arg \min_{x \in \R^d} \sum_{i \in S} f^\dagger_i(x). \label{eqn:2_out}
\end{align}
From~\eqref{eqn:2_f_dag}, 
\[\sum_{i \in S} f^\dagger_i(x) = \sum_{i \in S} \left( f_i(x) - \min_{y \in \R^d} f_i(y)\right) = \sum_{i \in S} f_i(x) - \sum_{i \in S} \, \min_{y \in \R^d} f_i(y).\]
This implies that $$\arg \min_{x \in \R^d} \sum_{i \in S} f^\dagger_i(x) = \arg \min_{x \in \R^d} \sum_{i \in S} f_i(x).$$ 
Substituting this in~\eqref{eqn:2_out} we obtain,
\[\widehat{x} \in \arg \min_{x \in \R^d} \sum_{i \in S} f_i(x).\]
The above argument holds for every execution of $\Pi^+$. Hence, by Definition~\ref{def:t_res}, the normalized implementation of $\Pi$ is {\em $t$-resilient}.
\end{proof}

% \printcommentNiru{ ++ note that t-res = (0, t)-wr +++++ lemma above implies (t, t)-wr = (0, t)-wr when costs have min zero and 2t-red ++++ therefore lemma 3 implies normalized implementation of (t,t)-wr = (0, t)-wr when costs are non-negative and have 2t-red +++ }

We have the following corollary of Theorem~\ref{thm:algo} and Theorem~\ref{thm:weak_strong}.

\begin{corollary}
\label{cor:tight}
If the true cost functions of the agents satisfy the {\em $2t$-redundancy} property, and are non-negative, then the {\em normalized implementation} of the proposed {\em $(t,\, t)$-weak resilient} algorithm in Section~\ref{sec:algo} is {\em $t$-resilient}. 
\end{corollary}

Note that the algorithm presented in this section is computationally much simpler than the {\em $t$-resilient} algorithm previously presented in Section~\ref{sub:t_algo}. However, the algorithm in this section relies on an additional assumption that the true cost function of each non-faulty agent is non-negative. In general, there is a trade-off between complexity of the algorithm, and the assumptions made regarding the true cost functions, as the discussion below also illustrates.

% Having said that, the additional assumption on the true cost functions is not contrived, and can be realized easily if the non-faulty agents implement a simple local pre-processing step before their cost functions over to the server. In particular, before sending its cost to the server, each non-faulty agent $i$ subtracts the minimum value of its true cost function, i.e., $\min_x f_i(x)$, from its true cost function $f_i(x)$ at all points $x$.

% Finally, we show that if the agents satisfy a redundancy property, that is much weaker than the {\em $2t$-redundancy} property, then the {\em sub-optimality} of a {\em weak resilient} algorithm can be significantly improved.

% %%%%%%%%%%%%%%%%%%%%%%%%%%%%%%%% REDUNDANCY VS RESILIENCE %%%%%%%%%%%%%%%%%%%%%%%%%%%%%%%%%%%
% \input{red_vs_res}

%%%%%%%%%%%%%%%%%%%%%%%%%%%%%%%%% DISTRIBUTED VERSION %%%%%%%%%%%%%%%%%%%%%%%%%%%%%%%%%%%%%%%
% \input{dist_algo}
\section{Gradient-Descent Based Algorithm}
\label{sec:norm}

In certain application of collaborative optimization, the algorithms only use information about the gradients of the agents' cost functions. Collaborative learning is one such application~\cite{bottou2018optimization}. Due to its practical importance, fault-tolerance in collaborative learning has gained significant attention in recent years~\cite{alistarh2018byzantine, bernstein2018signsgd, blanchard2017machine, chen2017distributed, xie2018generalized}. \\

In this section, we briefly summarize a gradient-descent based distributed collaborative optimization algorithm wherein the agents only send gradients of their cost functions to the server, instead of sending their entire cost functions. The algorithm was proposed in our prior work~\cite{gupta2019byzantine}, where we proved {\em $t$-resilience} of the algorithm when the true cost functions satisfy the {\em $2t$-redundancy} and certain additional properties.\\

% As elaborated in Section~\ref{sec:learn}, our algorithm can be easily modified for fault-tolerance in distributed learning. The computation cost of our proposed fault-tolerance mechanism (i.e., the CGC gradient-filter) is $\mathcal{O}(n (d + \log n))$, which is much lower than that of the fault-tolerance mechanisms used in many of the prior works~\cite{alistarh2018byzantine, blanchard2017machine, chen2017distributed, cao2018distributed}. Also, our algorithm applies in more general learning settings that satisfy $2f$-redundancy.\\

% The distributed optimization algorithm using the proposed CGC filter is presented below. The fault-tolerance guarantee of the algorithm is derived in Section~\ref{sub:ft}. \\

The proposed algorithm is iterative. For an execution of the algorithm, let $S$ denote the set of non-faulty agents and suppose that the true cost functions of the agents are $f_1(x), \ldots, \, f_n(x)$. The server maintains an estimate of the minimum point,
which is updated in each iteration of the algorithm. The initial estimate, named $x^0$, is chosen arbitrarily by the server from $\R^d$. In iteration $s \in \{0, \, 1, \ldots\}$,
the server computes estimate $x^{s+1}$ in steps S1 and S2 as described below.\\

In Step S1, the server obtains from the agents the gradients of their local cost functions at $x^s$. A faulty
agent may send an arbitrary $d$-dimensional vector for its gradient. Each non-faulty agent $i \in S$ sends the gradient of its true cost function at $x^s$, i.e., $\nabla f_i(x^s)$. In Step S2, to mitigate the detrimental impact of such incorrect gradients, the algorithm uses a filter to ``robustify" the gradient aggregation
step. In particular, the gradients with the largest $t$ norms are ``clipped'' so that their norm equals the norm of the $(t+1)$-th largest gradient (or, equivalently, the $(n-t)$-th smallest gradient). The remaining gradients remain unchanged.
The resulting gradients are then accumulated to obtain the update direction, which is then used to compute $x^{t+1}$. We refer to the method used in Step S2 for clipping the largest
$t$ gradients as ``{\em Comparative Gradient Clipping}" (CGC), since the largest $t$ gradients are clipped
to a norm that is ``comparable'' to the next largest gradient. \\

Detailed description of the algorithm and its resilience guarantee can be found in our prior work~\cite{gupta2019byzantine}.
The above algorithm performs correctly despite the use of a simple filter on the gradients, which only takes into account the gradient norms,
not the direction of the gradient vectors. This simplification is possible due to the assumptions made on the cost functions \cite{gupta2019byzantine}. Weaker assumptions will often necessitate more complex algorithms.

% The distributed optimization algorithm proposed is build upon the standard iterative gradient-descent method. Therefore, it is assumed that the non-faulty cost functions are as {\em smooth} as necessary. Upon presenting the algorithm, its convergence and fault-tolerance property are presented in Section~\ref{sub:ft}. \textcolor{blue}{For now, we assume the system to be synchronous. A partially asynchronous extension of the algorithm is presented later in Section~\ref{sec:extend}.} \\

% The server initiates the algorithm by choosing an arbitrary estimate $w^0$ of a {\em global} minimum point. In each iteration $t \in \{0, \, 1, \ldots\}$, the server computes the updated estimate $w^{t+1}$ using the following Steps {\bf S1} and {\bf S2}. The iterative process below can be terminated after a finite number of iterations, depending upon the desired accuracy and the step-size used.

%%%%%%%%%%%%%%%%%%%%%%%%%%%%%%%%% PRIOR WORK %%%%%%%%%%%%%%%%%%%%%%%%%%%%%%%%%%%%%%%
% \input{dist_algo}
% \input{prior}

%%%%%%%%%%%%%%%%%%%%%%%%%%%%%%%%% NON-DIFFERENTIABLE %%%%%%%%%%%%%%%%%%%%%%%%%%%%%%%%%%%%%%%
% \input{non_diff}

%%%%%%%%%%%%%%%%%%%%%%%%%%%%%%%%% SUMMARY %%%%%%%%%%%%%%%%%%%%%%%%%%%%%%%%%
\section{Summary of the Results}
\label{sec:sum}
We have made the following key contributions in this report.
\begin{itemize}
\setlength\itemsep{0.5em}
    \item {\bf In case of redundant cost functions:} We proved the necessary and sufficient condition of {\em $2t$-redundancy} for {\em $t$-resilience} in collaborative optimization. We have presented {\em $t$-resilient} collaborative optimization algorithms to demonstrate the trade-off between the complexity of a {\em $t$-resilient} algorithm, and the properties of the agents' cost functions. 
    
    \item  {\bf In case of independent cost functions:} We introduced the metric of {\em $(u, \, t)$-weak resilience} to quantify the notion of resilience in case when the agents' cost functions are independent. We have presented an algorithm that obtains {\em $(u, \, t)$-weak resilience} for all $u \geq t$ when the cost functions are non-negative and $n>2t$.
\end{itemize}

\section*{Acknowledgements}
Research reported in this paper was sponsored in part by the Army Research Laboratory under Cooperative Agreement W911NF- 17-2-0196, and by National Science Foundation award 1842198. The views and conclusions contained in this document are those of the authors and should not be interpreted as representing the official policies, either expressed or implied, of the the Army Research Laboratory, National Science Foundation or the U.S. Government.

% \input{minimax.tex}

% \input{stochastic.tex}
% %%%%%%%%%%%%%%%%%%%%%%%%%%%%%%%%% REDUNDANCY CASE %%%%%%%%%%%%%%%%%%%%%%%%%%%%%%%%%%%%%%%
% \input{redundancy.tex}

%%%%%%%%%%%%%%%%%%%%%%%%%%%%%%%%% NEED FOR 2T - REDUNDANCY %%%%%%%%%%%%%%%%%%%%%%%%%%%%%%%%%%%%%
% \input{app_need_redundancy.tex}

% \input{app_beta_gamma}

% \input{proof_algo}

% \input{app_red_bnd}

\bibliographystyle{plain}
\bibliography{ref.bib}

\begin{thebibliography}{10}

\bibitem{alistarh2018byzantine}
Dan Alistarh, Zeyuan Allen-Zhu, and Jerry Li.
\newblock Byzantine stochastic gradient descent.
\newblock In {\em Advances in Neural Information Processing Systems}, pages
  4618--4628, 2018.

\bibitem{bernstein2018signsgd}
Jeremy Bernstein, Jiawei Zhao, Kamyar Azizzadenesheli, and Anima Anandkumar.
\newblock signsgd with majority vote is communication efficient and {Byzantine}
  fault tolerant.
\newblock {\em arXiv preprint arXiv:1810.05291}, 2018.

\bibitem{bhatia2015robust}
Kush Bhatia, Prateek Jain, and Purushottam Kar.
\newblock Robust regression via hard thresholding.
\newblock In {\em Advances in Neural Information Processing Systems}, pages
  721--729, 2015.

\bibitem{blanchard2017machine}
Peva Blanchard, Rachid Guerraoui, Julien Stainer, et~al.
\newblock Machine learning with adversaries: {Byzantine} tolerant gradient
  descent.
\newblock In {\em Advances in Neural Information Processing Systems}, pages
  119--129, 2017.

\bibitem{bottou2018optimization}
L{\'e}on Bottou, Frank~E Curtis, and Jorge Nocedal.
\newblock Optimization methods for large-scale machine learning.
\newblock {\em Siam Review}, 60(2):223--311, 2018.

\bibitem{boyd2011distributed}
Stephen Boyd, Neal Parikh, Eric Chu, Borja Peleato, Jonathan Eckstein, et~al.
\newblock Distributed optimization and statistical learning via the alternating
  direction method of multipliers.
\newblock {\em Foundations and Trends{\textregistered} in Machine learning},
  3(1):1--122, 2011.

\bibitem{boyd2004convex}
Stephen Boyd and Lieven Vandenberghe.
\newblock {\em Convex optimization}.
\newblock Cambridge university press, 2004.

\bibitem{charikar2017learning}
Moses Charikar, Jacob Steinhardt, and Gregory Valiant.
\newblock Learning from untrusted data.
\newblock In {\em Proceedings of the 49th Annual ACM SIGACT Symposium on Theory
  of Computing}, pages 47--60, 2017.

\bibitem{chen2018resilient}
Yuan Chen, Soummya Kar, and Jose~MF Moura.
\newblock Resilient distributed estimation through adversary detection.
\newblock {\em IEEE Transactions on Signal Processing}, 66(9):2455--2469, 2018.

\bibitem{chen2017distributed}
Yudong Chen, Lili Su, and Jiaming Xu.
\newblock Distributed statistical machine learning in adversarial settings:
  {Byzantine} gradient descent.
\newblock {\em Proceedings of the ACM on Measurement and Analysis of Computing
  Systems}, 1(2):44, 2017.

\bibitem{chong2015observability}
Michelle~S Chong, Masashi Wakaiki, and Joao~P Hespanha.
\newblock Observability of linear systems under adversarial attacks.
\newblock In {\em American Control Conference}, pages 2439--2444. IEEE, 2015.

\bibitem{duchi2011dual}
John~C Duchi, Alekh Agarwal, and Martin~J Wainwright.
\newblock Dual averaging for distributed optimization: Convergence analysis and
  network scaling.
\newblock {\em IEEE Transactions on Automatic control}, 57(3):592--606, 2011.

\bibitem{gupta2019byzantine}
Nirupam Gupta and Nitin~H Vaidya.
\newblock Byzantine fault tolerant distributed linear regression.
\newblock {\em arXiv preprint arXiv:1903.08752}, 2019.

\bibitem{kairouz2019advances}
Peter Kairouz, H~Brendan McMahan, Brendan Avent, Aur{\'e}lien Bellet, Mehdi
  Bennis, Arjun~Nitin Bhagoji, Keith Bonawitz, Zachary Charles, Graham Cormode,
  Rachel Cummings, et~al.
\newblock Advances and open problems in federated learning.
\newblock {\em arXiv preprint arXiv:1912.04977}, 2019.

\bibitem{lamport1982byzantine}
Leslie Lamport, Robert Shostak, and Marshall Pease.
\newblock The {Byzantine} generals problem.
\newblock {\em ACM Transactions on Programming Languages and Systems (TOPLAS)},
  4(3):382--401, 1982.

\bibitem{lindell2010introduction}
Yehuda Lindell.
\newblock Introduction to coding theory lecture notes.
\newblock {\em Department of Computer Science Bar-Ilan University, Israel
  January}, 25, 2010.

\bibitem{lynch1996distributed}
Nancy~A Lynch.
\newblock {\em Distributed algorithms}.
\newblock Elsevier, 1996.

\bibitem{nedic2009distributed}
Angelia Nedic and Asuman Ozdaglar.
\newblock Distributed subgradient methods for multi-agent optimization.
\newblock {\em IEEE Transactions on Automatic Control}, 54(1):48--61, 2009.

\bibitem{pajic2017attack}
Miroslav Pajic, Insup Lee, and George~J Pappas.
\newblock Attack-resilient state estimation for noisy dynamical systems.
\newblock {\em IEEE Transactions on Control of Network Systems}, 4(1):82--92,
  2017.

\bibitem{pajic2014robustness}
Miroslav Pajic, James Weimer, Nicola Bezzo, Paulo Tabuada, Oleg Sokolsky, Insup
  Lee, and George~J Pappas.
\newblock Robustness of attack-resilient state estimators.
\newblock In {\em ICCPS'14: ACM/IEEE 5th International Conference on
  Cyber-Physical Systems (with CPS Week 2014)}, pages 163--174. IEEE Computer
  Society, 2014.

\bibitem{rabbat2004distributed}
Michael Rabbat and Robert Nowak.
\newblock Distributed optimization in sensor networks.
\newblock In {\em Proceedings of the 3rd international symposium on Information
  processing in sensor networks}, pages 20--27, 2004.

\bibitem{raffard2004distributed}
Robin~L Raffard, Claire~J Tomlin, and Stephen~P Boyd.
\newblock Distributed optimization for cooperative agents: Application to
  formation flight.
\newblock In {\em 2004 43rd IEEE Conference on Decision and Control (CDC)(IEEE
  Cat. No. 04CH37601)}, volume~3, pages 2453--2459. IEEE, 2004.

\bibitem{shoukry2015imhotep}
Yasser Shoukry, Pierluigi Nuzzo, Alberto Puggelli, Alberto~L
  Sangiovanni-Vincentelli, Sanjit~A Seshia, Mani Srivastava, and Paulo Tabuada.
\newblock Imhotep-smt: A satisfiability modulo theory solver for secure state
  estimation.
\newblock In {\em Proc. Int. Workshop on Satisfiability Modulo Theories}, 2015.

\bibitem{shoukry2017secure}
Yasser Shoukry, Pierluigi Nuzzo, Alberto Puggelli, Alberto~L
  Sangiovanni-Vincentelli, Sanjit~A Seshia, and Paulo Tabuada.
\newblock Secure state estimation for cyber-physical systems under sensor
  attacks: A satisfiability modulo theory approach.
\newblock {\em IEEE Transactions on Automatic Control}, 62(10):4917--4932,
  2017.

\bibitem{su2018finite}
Lili Su and Shahin Shahrampour.
\newblock Finite-time guarantees for {Byzantine}-resilient distributed state
  estimation with noisy measurements.
\newblock {\em arXiv preprint arXiv:1810.10086}, 2018.

\bibitem{su2016fault}
Lili Su and Nitin~H Vaidya.
\newblock Fault-tolerant multi-agent optimization: optimal iterative
  distributed algorithms.
\newblock In {\em Proceedings of the 2016 ACM symposium on principles of
  distributed computing}, pages 425--434. ACM, 2016.

\bibitem{su2016robust}
Lili Su and Nitin~H Vaidya.
\newblock Robust multi-agent optimization: coping with {Byzantine} agents with
  input redundancy.
\newblock In {\em International Symposium on Stabilization, Safety, and
  Security of Distributed Systems}, pages 368--382. Springer, 2016.

\bibitem{sundaram2018distributed}
Shreyas Sundaram and Bahman Gharesifard.
\newblock Distributed optimization under adversarial nodes.
\newblock {\em IEEE Transactions on Automatic Control}, 2018.

\bibitem{van1971coding}
Jacobus~Hendricus Van~Lint.
\newblock {\em Coding theory}, volume 201.
\newblock Springer, 1971.

\bibitem{xie2018generalized}
Cong Xie, Oluwasanmi Koyejo, and Indranil Gupta.
\newblock Generalized {Byzantine}-tolerant sgd.
\newblock {\em arXiv preprint arXiv:1802.10116}, 2018.

\bibitem{yang2017byrdie}
Zhixiong Yang and Waheed~U. Bajwa.
\newblock Byrdie: {Byzantine}-resilient distributed coordinate descent for
  decentralized learning, 2017.

\end{thebibliography}

%%%%%%%%%%%%%%%%%%%%%%%%%%%%%%%%% APPENDIX %%%%%%%%%%%%%%%%%%%%%%%%%%%%%%%%%%%%%
% \newpage

\appendix
\section{Proof of Lemma \ref{lem:non-empty-x}}
\label{append:non-empty-x}

\noindent{\bf Lemma \ref{lem:non-empty-x}.} {\em 
% Suppose that Assumption~\ref{asp:basic} holds. 
For a non-empty set $T$, consider a set of functions $g_i(x)$, $i\in T$, such that $$\bigcap_{i\in T}\arg\min_x g_i(x)\neq \emptyset.$$ Then
$$\bigcap_{i\in T}\arg\min_x g_i(x) = \arg\min_x \sum_{i\in T}g_i(x).$$
} \\

\begin{proof}
Consider any non-empty set $T$, and functions $g_i(x)$, $i\in T$, such that $$\bigcap_{i\in T}\arg\min_x g_i(x)\neq \emptyset.$$

{\bf Part I:} Consider any $x^o\in \bigcap_{i \in T} \arg \min_x ~ g_i(x)$.
Since each cost function $g_i(x)$, $i\in T$, is minimized at $x^o$, it
follows that $\sum_{i \in T} g_i(x)$ is also minimized at $x^o$. In other words,
it is trivially true that
\begin{align}
    x^o \in \bigcap_{i \in T} \arg \min_x ~ g_i(x) \subseteq \arg \min_x \sum_{i \in T} g_i(x).\label{eqn:lem:1} 
\end{align}

~

{\bf Part II:} Let $x^-$ be a point such that
$$x^-\in \bigcap_{i \in T} \arg \min_x ~ g_i(x).$$
Then
\begin{align}
 x^- \in\arg \min_x ~ g_i(x), ~~~~\forall\,i\in T
 \label{eqn:argmin_i:lem}
\end{align}

From~\eqref{eqn:lem:1}, $\arg \min_x \sum_{i \in T} g_i(x)\neq\emptyset$.
Now we show that
$\arg \min_x \sum_{i \in T} g_i(x) \subseteq \bigcap_{i \in T} \arg \min_x ~ g_i(x)$.
The proof is by contradiction. \\

Suppose that there exists
a point $x^{\dagger}$ such that
\begin{align}
    x^{\dagger} \in \arg \min_x \sum_{i \in T} g_i(x), \label{eqn:x_dag}
\end{align}
and 
\begin{align}
x^{\dagger} \not \in \bigcap_{i \in T} \arg \min_x ~ g_i(x).
\end{align}
This and \eqref{eqn:argmin_i:lem} implies that there exists $i^{\dagger} \in T$ such that 
\[g_{i^{\dagger}} (x^{\dagger}) > g_{i^{\dagger}} (x^-).\]
Also, from~\eqref{eqn:argmin_i:lem}, for each $i \in T \setminus \{i^{\dagger}\}$, 
\[g_i(x^{\dagger}) \geq g_i(x^-).\] 
The above two inequalities together imply that,
\begin{align*}
    \sum_{i \in T} g_i(x^{\dagger}) = \sum_{i \in T \setminus \{i^{\dagger}\}}g_i(x^{\dagger}) + g_{i^{\dagger}}(x^{\dagger}) > \sum_{i \in T}g_i(x^-).
\end{align*}
The above is a contradiction of~\eqref{eqn:x_dag}. Therefore, $x^{\dagger} \in \bigcap_{i \in T} \arg \min_x ~ g_i(x)$. This implies,
\begin{align}
    \arg \min \sum_{i \in T} g_i(x) \subseteq \bigcap_{i \in T} \arg \min ~ g_i(x).  \label{eqn:in_1}
\end{align}
From~\eqref{eqn:lem:1} and~\eqref{eqn:in_1},
\begin{align}
    \arg \min \sum_{i \in T} g_i(x) = \bigcap_{i \in T} \arg \min ~ g_i(x)  
\end{align}

\end{proof}

\section{Definitions~\ref{def:2t_red_alt} and~\ref{def:2t_red} are Equivalent}
\label{app:equiv_def}

The lemma below shows that the two definitions of {\em $2t$-redundancy}, namely Definition~\ref{def:2t_red_alt} and Definition~\ref{def:2t_red}, stated in Section~\ref{sub:res} are equivalent.

\begin{lemma}
\label{lem:equiv}
Suppose that Assumption~\ref{asp:basic} holds true, and $n > 2t$. Then, conditions in Definition~\ref{def:2t_red_alt} and Definition~\ref{def:2t_red} are equivalent.
\end{lemma}
\begin{proof}
% For $t=0$, the proof immediately follows from Assumption~\ref{asp:basic} and Lemma \ref{lem:non-empty-x}. Thus, we assume $t>0$ in the remaining proof. \\
% \printcomment{+++++++ I added the above, please check for correctness. ++++ do we address t=0 below too? then this is redundant here ++++}
%
Let the true cost functions of each agent $i$ be denoted by $f_i(x)$.\\ 

% \printcomment{+++++++ does this require assumption 1? +++++++ then lemma should state it. In particular, the proof for t=0 follows by Lemma 1, not immediate as such ++++++ {\color{blue} Corrected}}

{\bf Part I:} We first show that the condition in
Definition~\ref{def:2t_red_alt} implies that in Definition~\ref{def:2t_red}. Suppose that the condition stated in
Definition~\ref{def:2t_red_alt} holds true.

% Now, since the set intersections in equation (\ref{def_1}) are non-empty, it implies the following, where the first intersection below is over all subsets $S_i \subset S$ each containing $n-2t$ agents.
% \begin{align}
% \bigcap_{S_i,|S_i|=n-2t}\left(\bigcap_{i \in S_i} \arg \min f_i(x)\right)~\neq~\emptyset
% \label{e_def_11}
% \end{align}

% \printcomment{++++ next sentence assumes Si is subset of S. But definition 2 does not say that. {\color{blue} Corrected.}}

% \printcomment{ ++++++ Also the text refers to definition 4 here, which is incorrect ++++++++++{\color{blue}Correction made.}}

% \printcommentNiru{ ++++ Couple of things were incorrectly concluded. I have made the corrections.}

Consider two arbitrary sets of agents $S$ and $\widehat{S}$ with $\mnorm{S} \geq n-t$, $\mnorm{\widehat{S}} \geq n-2t$, and $\widehat{S} \subseteq S$. We need to show that~\eqref{def_2} in Definition~\ref{def:2t_red} holds true.\\

Note,~\eqref{def_1} in Definition~\ref{def:2t_red_alt} implies that there exists a point $x^*$ such that
\begin{align*}
    x^* \in \bigcap_{i \in S^\dagger} \arg \min_{x \in \R^d} f_i(x) , \quad \forall \, S^\dagger \subseteq S, ~ \mnorm{S^\dagger} = n-2t.
\end{align*}
Therefore, 
\[x^* \in \bigcap_{i \in S} \arg \min f_i(x) \neq \emptyset.\]
% Now, observe that $\bigcup S_i=S$, the set of all non-faulty agents as defined in Definition \ref{def:2t_red_alt}.
% Then we have $S=\bigcup S_i$, and (\ref{e_def_11}) implies that $$\bigcap_{i \in S} \arg \min f_i(x)\neq \emptyset.$$ 
Thus, from Lemma~\ref{lem:non-empty-x},
\begin{align}
    \bigcap_{i \in S} \arg \min ~ f_i(x) = \arg \min \sum_{i \in S} f_i(x)
    \label{e_p_0}
\end{align}
Now, consider an arbitrary subset $S_1 \subseteq \widehat{S}$ with $\mnorm{S_1} = n-2t$. Then, 
\begin{align}
    \bigcap_{i \in S} \arg \min ~ f_i(x) \subseteq \bigcap_{i \in \widehat{S}} \arg \min ~ f_i(x) \subseteq \bigcap_{i \in S_1} \arg \min ~ f_i(x). \label{eqn:alt_s_s1}
\end{align}
We now show that when the condition in Definition~\ref{def:2t_red_alt} holds true then $\bigcap_{i \in S} \arg \min ~ f_i(x) = \bigcap_{i \in S_1} \arg \min ~ f_i(x)$.
The proof is by contradiction.\\

Suppose that 
\begin{align}
    \bigcap_{i \in S} \arg \min ~ f_i(x) \subset \bigcap_{i \in S_1} \arg \min ~ f_i(x).  \label{eqn:alt_cont}
\end{align}
This implies that there exists a point $x^\dagger$ in $\bigcap_{i \in S_1} \arg \min ~ f_i(x)$ such that 
\[x^\dagger \not \in \bigcap_{i \in S} \arg \min ~ f_i(x).\]
This implies that there exists $i^\dagger \in S$ such that $x^\dagger \not \in \arg \min f_{i^\dagger}(x)$. Now, consider a subset  $S_2\subseteq S$ with $|S_2|=n-2t$ and $i^\dagger\in S_2$. Then,
\[x^\dagger \not \in \bigcap_{i \in S_2} \arg \min ~ f_i(x).\]
Since $x^\dagger\in\bigcap_{i \in S_1} \arg \min ~ f_i(x)$, the above implies that 
\[
\bigcap_{i \in S_1} \arg \min ~ f_i(x) \neq \bigcap_{i \in S_2} \arg \min ~ f_i(x)
\]
which contradicts~\eqref{def_1} in Definition~\ref{def:2t_red_alt}. Therefore,~\eqref{eqn:alt_cont} cannot hold, and so,
\begin{align}
    \bigcap_{i \in S} \arg \min ~ f_i(x) = \bigcap_{i \in S_1} \arg \min ~ f_i(x).
    \label{e_p_2}
\end{align}
Substituting the above in~\eqref{eqn:alt_s_s1} implies that
\begin{align*}
    \bigcap_{i \in S} \arg \min ~ f_i(x) = \bigcap_{i \in \widehat{S}} \arg \min ~ f_i(x)
\end{align*}
The above together with~\eqref{e_p_0} imply that
\begin{align}
    \bigcap_{i \in \widehat{S}} \arg \min ~ f_i(x) = \arg \min \sum_{i \in S} f_i(x). \label{eqn:alt_s=s1}
\end{align}
% For any set $\widehat{S}$ containing $n-2t$ non-faulty nodes, by defining $S_1=\widehat{S}$, the above equation implies that the condition in Definition \ref{def:2t_red} holds.\\

% Recall that $S_1\subseteq \widehat{S}\subseteq S$. Therefore, 
% \[\bigcap_{i \in S} \arg \min ~ f_i(x) \subseteq \bigcap_{i \in \widehat{S}} \arg \min ~ f_i(x) \subseteq \bigcap_{i \in S_1} \arg \min ~ f_i(x).\]
% Substituting the above in~\eqref{eqn:alt_s=s1} implies,
% \begin{align*}
%     \bigcap_{i \in S} \arg \min ~ f_i(x) = \bigcap_{i \in \widehat{S}} \arg \min ~ f_i(x) \subseteq \widehat{S}. 
% \end{align*}
Note that the above argument holds true for all pairs of sets $\widehat{S}, \, S$ with $\mnorm{\widehat{S}} \geq n-2t$. $\mnorm{S} \geq n-t$, and $\widehat{S} \subseteq S$. Therefore, the above implies that the condition stated in Definition~\ref{def:2t_red} is true.\\

{\bf Part II:}
We now show that the condition in Definition \ref{def:2t_red} implies the condition in Definition \ref{def:2t_red_alt}. Suppose that the condition stated in Definition~\ref{def:2t_red} holds true.\\

From Assumption~\ref{asp:basic}, 
\[\arg \min \sum_{i  = 1}^n f_i(x) \neq \emptyset.\]
From substituting $S = \{1, \ldots, \, n\}$ in the equation~\eqref{def_2} in Definition~\ref{def:2t_red}, we trivially obtain the following for every two subsets of agents $S_1$ and $S_2$ each containing $n-2t$ agents.
\[\emptyset \neq \bigcap_{i \in S_1} \arg \min f_i(x) = \bigcap_{i \in S_2} \arg \min f_i(x).\]
Hence, the condition in Definition~\ref{def:2t_red} implies the condition in Definition~\ref{def:2t_red_alt}. 
\end{proof}

\end{document}